\documentclass{article}

\usepackage[final]{neurips_2020}

\usepackage[utf8]{inputenc} 
\usepackage[T1]{fontenc}    
\usepackage{hyperref}       
\usepackage{url}            
\usepackage{booktabs}       
\usepackage{amsfonts}       
\usepackage{nicefrac}       
\usepackage{microtype}      
\usepackage{amsmath}
\usepackage{amsthm}
\usepackage{amssymb}
\usepackage{dsfont}
\usepackage{mathtools}
\usepackage{enumerate}

\usepackage[linesnumbered,ruled]{algorithm2e}
\IncMargin{1.5em}

\usepackage{cleveref}
\usepackage{caption}
\usepackage{subcaption}
\usepackage{graphicx}
\usepackage{float}
\usepackage{multirow}
\usepackage{booktabs}
\usepackage{multicol}

\newtheorem{corollary}{Corollary}[section]
\newtheorem{definition}{Definition}[section]

\newtheorem{lemma}{Lemma}[section]

\newtheorem{remark}{Remark}[section]
\newtheorem{theorem}{Theorem}[section]
\newtheorem{assumption}{Assumption}[section]

\DeclareMathOperator*{\argmax}{arg\,max}
\DeclareMathOperator*{\dom}{dom}
\newcommand{\var}{\textnormal{VaR}_{\alpha}}
\newcommand{\cvar}{\textnormal{CVaR}_{\alpha}}
\newcommand{\Prob}{\mathbb{P}}
\newcommand{\Egpd}{\epsilon_{u,\alpha}}
\newcommand{\hEgpd}{\hat{\epsilon}^{(n)}_{\alpha}}
\newcommand{\hXi}{\hat{\xi}^{\scriptscriptstyle{(n)}}_{\scriptscriptstyle{\textnormal{MLE}}}}
\newcommand{\hRho}{\hat{\rho}_n}
\newcommand{\hSig}{\hat{\sigma}^{\scriptscriptstyle{(n)}}_{\scriptscriptstyle{\textnormal{MLE}}}}
\newcommand{\UPOT}{\hat{c}^{(n)}_{\epsilon,\alpha}}
\newcommand{\POT}{\hat{c}^{(n)}_{\alpha}}
\newcommand{\pto}{\overset{p}{\to}}
\newcommand{\dto}{\overset{d}{\to}}
\newcommand{\thresh}{X_{(n-k,n)}}

\usepackage[dvipsnames]{xcolor}
\definecolor{MyDarkBlue}{rgb}{0,0.08,0.45}
\def\boxit#1{\vbox{\hrule\hbox{\vrule\kern6pt
          \vbox{\kern6pt#1\kern6pt}\kern6pt\vrule}\hrule}}

\title{Bias-Corrected Peaks-Over-Threshold Estimation of the CVaR}

\author{%
Dylan Troop\thanks{Corresponding author.}\\
  Department of Computer Science and Software Engineering\\
  Concordia University\\
  Montreal, Canada \\
  \texttt{d\_troop@encs.concordia.ca} \\
  \And
  Fr\'ed\'eric  Godin\thanks{Quantact Actuarial and Financial Mathematics
  Laboratory, Montreal, Canada.} \\
  Department of Mathematics and Statistics \\
  Concordia University\\
  Montreal, Canada \\
  \texttt{frederic.godin@concordia.ca} \\
  \And
  Jia Yuan Yu \\
  Concordia Institute of Information Systems Engineering \\
  Concordia University \\
  Montreal, Canada \\
  \texttt{jiayuan.yu@concordia.ca} \\
}

\begin{document}

\maketitle

\begin{abstract}
The conditional value-at-risk (CVaR) is a useful risk measure in fields such as machine learning, finance, insurance, energy, etc. When measuring very extreme risk, the commonly used CVaR estimation method of sample averaging does not work well due to limited data above the value-at-risk (VaR), the quantile corresponding to the CVaR level. To mitigate this problem, the CVaR can be estimated by extrapolating above a lower threshold than the VaR using a generalized Pareto distribution (GPD), which is often referred to as the peaks-over-threshold (POT) approach. This method often requires a very high threshold to fit well, leading to high variance in estimation, and can induce significant bias if the threshold is chosen too low. In this paper, we derive a new expression for the GPD approximation error of the CVaR, a bias term induced by the choice of threshold, as well as a bias correction method for the estimated GPD parameters. This leads to the derivation of a new estimator for the CVaR that we prove to be asymptotically unbiased. In a practical setting, we show through experiments that our estimator provides a significant performance improvement compared with competing CVaR estimators in finite samples. As a consequence of our bias correction method, it is also shown that a much lower threshold can be selected without introducing significant bias. This allows a larger portion of data to be be used in CVaR estimation compared with the typical POT approach, leading to more stable estimates. As secondary results, a new estimator for a second-order parameter of heavy-tailed distributions is derived, as well as a confidence interval for the CVaR which enables quantifying the level of variability in our estimator.
\end{abstract}

\section{Introduction} \label{se:intro}
Traditional machine learning algorithms typically consider the expected value of a random variable as the target to optimize. In a risk-averse setting, the objective function needs to be adapted to consider the full distribution and account for severe outcomes. Recently, risk-averse machine learning has become an important area of study, especially in the context of multi-armed bandits and reinforcement learning, for example, for example, \cite{NIPS2014_mdp, tamar2015, torossian2019xarmed, NIPS2019_cvar} and \cite{keramati2019optimistic} address learning with a risk-averse agent. Most often, the risk measure of interest is the \textit{conditional value-at-risk} (CVaR). Given a continuous random variable $X$ representing losses (i.e., where larger values are less desirable), the CVaR at a confidence level $\alpha \in (0,1)$ measures the expected value of $X$ given that $X$ exceeds the quantile of level $\alpha.$ This quantile is referred to as the \textit{value-at-risk} (VaR). Compared to the VaR, the CVaR captures more information about the weight of a distribution's tail, making it a more useful object of study in risk-averse decision making. In practice, the CVaR is usually estimated by averaging observations above the estimated VaR, which we call the \textit{sample average} estimator of the CVaR. When $\alpha$ is close to $1$, these observations can be very scarce in small samples leading to volatile estimates of the CVaR. This work is motivated by a lack of reliable estimators and performance guarantees for the CVaR at these extreme levels. In this paper, we consider estimating the CVaR of \textit{heavy-tailed random variables}, which are ubiquitous in application domains such as finance, insurance, energy, and epidemiology, e.g., \citet{Manz2020}. In this setting, extreme events correspond to very large observations (and hence severe losses), which is in contrast to the light- or short-tailed cases where similar low probability events are closer to the mean. Extreme value theory provides the tools to construct a new CVaR estimator that is appropriate for this setting. By selecting a threshold lower than the VaR, it is possible to approximate the tail distribution of a random variable by using a \textit{generalized Pareto distribution} (GPD) and extrapolating beyond available observations. The estimation of quantities using this approximation is commonly referred to as the \textit{peaks-over-threshold} (POT) approach. An estimator for the CVaR using the POT approach is given in, for example, \citet[Section 7.2.3]{mcneil2005quantitative}, where the CVaR is referred to as the \textit{expected shortfall}. However, this estimator suffers from one of the main drawbacks of the POT approach, which is the difficult bias-variance tradeoff in selecting the threshold. Unless the threshold is chosen very high, the estimator will encounter two sources of potentially significant bias:
\begin{enumerate}[(i)]
\item The deviation between the GPD and the true tail distribution;
\item The bias associated with parameter estimation  using the approximate GPD tail data.
\end{enumerate}
Perhaps even more significantly, the CVaR estimator of \cite{mcneil2005quantitative} comes with no performance guarantees unless one assumes exactness of the GPD approximation and of the empirical distribution function, so it has not been previously possible to determine the precise conditions where using the POT approach for CVaR estimation is actually superior to the more common sample average estimator. The goal of this paper is to make a significant refinement to the existing CVaR estimator based on the POT approach and correct the sources of bias induced by the GPD approximation, resulting in a more accurate estimator that is less sensitive to the choice of threshold, as well as to derive performance guarantees in the form of confidence intervals.

\textbf{Related work}. Threshold selection methods for the GPD approximation have been well-studied, for an overview of such methods, see \cite{scarrott2012review}. A recent and empirically successful threshold selection method is given in \cite{bader2018automated}, and has been applied in \cite{extreme_tail_risk_pot} to estimate the VaR at extreme levels using the POT approach. Estimation of the CVaR using the POT approach has been applied in, for example, \cite{Gilli2006, Marinelli2007, Bah2016, GKILLAS2018109} and \cite{Szubzda2019}. To the best of our knowledge, results in the literature using the POT approach for risk estimation are presented mostly in the form of empirical studies. On the theoretical side, asymptotic analysis of the deviation in (i) with respect to the underlying distribution can be found in \cite{raoult_worms_2003}, and \cite{Beirlant2003OnTR} considers the deviation with respect to quantile approximations based on the GPD tail model. Bias correction methods for (ii) have been developed in, for example, \cite{PENG1998107,GOMESA2002,BEIRLANT20092800} and \cite{nawel}. These works target the estimation of $\xi$, the shape parameter of the GPD. A central idea is to employ the theory of second-order regular variation to establish asymptotically unbiased extensions of the well-known Hill estimator of $\xi$. However, maximum likelihood estimation is considered the most efficient parameter estimation method for the GPD (see \citet[Section 7]{BERMUDEZ}) and targets both $\xi$ and the $\sigma$ (the scale parameter of the GPD). This paper is, to the best of our knowledge, the first to address (ii) for maximum likelihood estimation in view of application. In terms of performance guarantees, \textit{concentration bounds} for the CVaR estimated by sample averaging exist in the literature, which measure the probability of deviation between the CVaR estimate and its true value for a given sample size. While a major benefit with concentration bounds is that they provide a guaranteed bound in finite samples, it is usually not possible or impractical to apply them to heavy-tailed random variables. For example, \cite{brown,thomasandmiller} derive concentration bounds which apply to bounded random variables. The concentration bounds of \cite{ravi_unbounded,ravi_tails,NIPS2019_9347} can be used for heavy-tailed random variables but require distribution-dependent constants, making them impractical. \cite{NIPS2019_9305} give concentration bounds which can apply to heavy-tailed random variables without exact distributional knowledge, but their bound is based on a truncated version of the sample average CVaR estimator and requires selecting parameters based on moment bounds of the underlying random variable. An alternative to concentration bounds is to use estimated \textit{asymptotic confidence intervals}, which typically become good approximations in large sample sizes. Asymptotic confidence intervals for the CVaR using sample average estimation can be found in \cite{TRINDADE20073524, BRAZAUSKAS20083590} and \cite{SUN2010246}. These confidence intervals only apply to random variables with bounded variance, which excludes many heavy-tailed random variables. Therefore, it is often not possible to establish measures of uncertainty for CVaR estimates using either the sample average method or the original POT method in the heavy-tailed domain, and this paper aims to address this problem with a rigorous formulation of a new CVaR estimator using the POT approach. 

\textbf{Major contributions}. First, we derive the \textit{GPD approximation error}, a deterministic quantity measuring the deviation between the GPD approximation of the CVaR and its true value. We then derive bias-corrected maximum likelihood estimators for the GPD parameters, $\xi$ and $\sigma$, using the POT approach, which in turn requires the derivation of a new estimator for a second-order parameter that may be of independent interest. Using our bias correction methods, a new estimator for the CVaR based on the POT approach is derived which we prove is asymptotically unbiased. Using our convergence result for the bias-corrected CVaR estimator, we derive a confidence interval for the CVaR which has asymptotically correct coverage probability.

The remainder of this paper is organized as follows. In \cref{se:prelim}, the VaR and CVaR are formally defined, and the sample average estimator of the CVaR is given. Needed background from extreme value theory and second-order regular variation is discussed and we formalize the notion of heavy-tailed random variables. The CVaR approximation using the POT approach is given. \Cref{se:gpd_approx_error} derives the GPD approximation error for the CVaR and its asymptotic behaviour. In \cref{se:pot_mle}, bias-corrected maximum likelihood estimators for the GPD parameters using the POT approach are derived. These estimators are then used in a CVaR estimator with partial bias correction and its asymptotic normality is derived. \Cref{se:confidence_intervals} establishes an estimator for the GPD approximation error, and our results are consolidated to give the unbiased POT estimator for the CVaR. Its asymptotic normality is derived and a confidence interval is given. \Cref{se:second-order} gives details on second-order parameter estimation, which plays an important role in bias correction. In \cref{se:num_sims}, simulations are shown to provide empirical evidence of the finite sample performance of our estimator on data. \Cref{se:conclusion} concludes, with directions for future work. All proofs are given in \cref{appendix:proofs}.
\section{Preliminaries} \label{se:prelim}
Let $X$ denote a random variable and $F$ its corresponding cumulative distribution function (cdf). In this paper, we adopt the convention that $X$ represents a loss, so larger values of $X$ are less desirable. 
\begin{definition}[Value-at-Risk]
The value-at-risk of $X$ at level $\alpha \in (0,1)$ is
\begin{equation} \label{var_formula}
q_\alpha \triangleq \var(X) = \inf\{x \in \mathbb{R} | F(x) \ge \alpha\}.
\end{equation}
\end{definition}
$\var(X)$ is equivalent to the quantile at level $\alpha$ of $F$. If the inverse of $F$ exists, $\var(X) = F^{-1}(\alpha)$. The VaR can be estimated in the same way as the standard empirical quantile. Let $X_1, \ldots, X_n$ be i.i.d. random variables with common cdf $F$. Let $X_{(1,n)}\leq X_{(2,n)}\ldots\leq X_{(n,n)}$ denote the set of order statistics for the sample of size $n$, i.e., the sample sorted in non-decreasing order. An estimator for the VaR is
\[
\widehat{\textnormal{VaR}}_{n,\alpha}(X) = \\\min\left\{X_{(i,n)} \, \vert \, i=1,\ldots,n; \, \hat{F}_n\left(X_{(i,n)}\right) \geq \alpha\right\} = X_{(m,n)},
\]
where $\hat{F}_n$ denotes the empirical cdf and $m = \lceil\alpha n\rceil$. We now define the CVaR as in \cite{acerbi2002coherence}.\footnote{The expression given in \cite{acerbi2002coherence} is for the \textit{expected shortfall} (ES), but they show that the ES and CVaR are equivalent. They also use different conventions for $X$ and $\alpha$, where smaller values of $X$ represent less desired outcomes and $\alpha$ represents a tail probability. Equation \eqref{cvar_formula} can be derived by replacing their $X$ with $-X$ and their $\alpha$ with $1-\alpha$ in the equations of the former paper.}
\begin{definition} [Conditional Value-at-Risk] \label{cvar_definition}
The conditional value-at-risk of a continuous random variable $X$ at level $\alpha \in (0,1)$ is \begin{equation} \label{cvar_formula}
c_\alpha \triangleq \cvar(X) = \mathbb{E}[X | X \geq \var(X)]=\frac{1}{1-\alpha}\int_{\alpha}^{1}\textnormal{VaR}_\gamma(X) d\gamma.
\end{equation}
\end{definition}
Typical values of $\alpha$ are $0.95$, $0.99$, $0.999$, etc. Without loss of generality, the current work will only consider continuous random variables. Typically, the CVaR is estimated by averaging observations above $\widehat{\textnormal{VaR}}_\alpha(X)$. This estimator is given by
 \begin{equation} \label{sa_cvar}
  \widehat{\textnormal{CVaR}}_{n,\alpha}(X) = \frac{\sum_{i=1}^{n} X_{i} \mathds{1}_{ \{X_i \geq \widehat{\textnormal{VaR}}_{n,\alpha}(X) \} }} {\sum_{j=1}^{n} \mathds{1}_{ \{X_j \geq \widehat{\textnormal{VaR}}_{n,\alpha}(X) \} }}.
\end{equation}
The use of the \cref{sa_cvar} can be problematic when the confidence level $\alpha$ is high due to the scarcity of extreme observations. We now provide tools from extreme value theory to address this problem, which will be needed to give the CVaR estimator based on the POT approach. 

Let
$F^n(x) = \mathbb{P}(\max(X_1, \ldots, X_n) \le x)$
denote the cdf of the sample maxima. Suppose there exists a sequence of real-valued constants $a_n > 0$ and $b_n$, $n=1,2,\ldots$, and a nondegenerate cdf $H$ such that
\begin{equation} \label{limit_dist}
\underset{n \rightarrow \infty}{\lim } F^n \left(a_n x + b_n\right) = H(x),
\end{equation}
for all $x$, where \textit{nondegenerate} refers to a distribution not concentrated at a single point. The class of distributions $F$ that satisfy \eqref{limit_dist} are said to be in the \textit{maximum domain of attraction of H}, denoted $F \in \textnormal{MDA}(H)$. The Fisher–Tippett–Gnedenko theorem (see \citet[Theorem 1.1.3]{dehaan2006}) states that $H$ must then be a \textit{generalized extreme value distribution} (GEVD), given in the following definition.
\begin{definition}[GEVD]
  The generalized extreme value distribution (GEVD) with single
  parameter $\xi\in\mathbb R$ has distribution function
  \begin{equation*}
    H_\xi(x) = \begin{cases}
      \exp \left(-(1+\xi x )^{-1/\xi}\right)
      \quad \text{ if } \xi \neq 0,
      \\ \exp \left(-e^{-x}\right)\quad \text{ if } \xi = 0
    \end{cases}
  \end{equation*}
  over its support, which is $[-1/\xi,\infty)$ if $\xi > 0$, $(-\infty,-1/\xi]$ if $\xi < 0$ or $\mathbb{R}$ if $\xi=0$.
\end{definition} 
If $F \in \textnormal{MDA}(H)$, then there exists a unique $\xi \in \mathbb{R}$ such that $H=H_\xi$. It is important to note that essentially all common distributions used in applications are in MDA$(H_\xi)$ for some value of $\xi$. When $\xi > 0$, $F$ is a \textit{heavy-tailed distribution}. It is useful to characterize heavy tails using the theory of regular variation, which requires the following definition.
\begin{definition}[Regularly varying function]
Let $f$ be a positive, measurable function defined on some neighborhood $[x_0, \infty)$ of $\infty$, for some $x_0 \in \mathbb{R}$. 
If
$$\lim _{x \rightarrow \infty} f(t x)/f(x)=t^{\rho} \quad \text { for all } t>0,$$
then $f$ is called regularly varying (at infinity) with unique \textit{index of regular variation} $\rho \in \mathbb{R}$, and we denote this by $f \in RV_\rho$. If $\rho=0$, then $f$ is called slowly-varying.
\end{definition}
For the remainder of this paper, we focus exclusively on heavy-tailed random variables (or distributions), defined next. We denote the tail distribution $\bar{F} = 1-F$.
\begin{definition}[Heavy-tailed random variable]
Let $X$ be a random variable with cdf $F$. Then $X$ (or $F$) is heavy-tailed if $F \in \textnormal{MDA}(H_\xi)$ with $\xi > 0$.
\end{definition}
If $F$ is heavy-tailed, then moments of order greater than or equal to $1/\xi$ do not exist. Otherwise, $F$ is light-tailed with a tail having exponential decay ($\xi=0$), or the right endpoint of $F$ is finite ($\xi < 0$). If $\xi \ge 1$, then $F$ has infinite mean, and therefore the true CVaR, \cref{cvar_formula}, is also infinite. For the remainder of this paper, we assume the following condition is satisfied.
\begin{assumption}\label{F_assumption}
$F$ is heavy-tailed with $\xi < 1$.
\end{assumption}
When $F \in \textnormal{MDA}(H_\xi)$, there exists a useful approximation of the distribution of sample extremes above a threshold, and we define this distribution next.
\begin{definition}[Excess distribution function] \label{def:edf}
For a given threshold
$u \ge \textnormal{ess inf} \, X$, the excess distribution function is defined as
\[
  F_{u}(y) = \mathbb{P}(X-u \leq y | X > u)
           = \left[F(y+u)-F(u)\right]/\bar{F}(u),
          \qquad y>0.
\]
\end{definition}
Note that the domain of $F_{u}$ is $[0,\infty)$ under \cref{F_assumption}. The $y$-values are referred to as the \textit{threshold excesses}. Given that $X$ has exceeded some high threshold $u$, this function represents the probability that $X$ exceeds the threshold by at most $y$. The Pickands-Balkema-de Haan theorem states that $F_u$ can be well-approximated by the GPD, which we give now.
\begin{theorem}[\cite{pickands1975, balkema1974l}] \label{PBH} Suppose \cref{F_assumption} is satisfied. Then, there exists a positive function $\sigma = \sigma(u)$ such that
  \begin{equation} \label{eq:PBH}
    \underset{u \rightarrow \infty}{\lim} \,\,\underset{0 \leq y \leq \infty}{\sup} \vert F_{u}(y) - G_{\xi, \sigma}(y) \vert = 0,
  \end{equation}
  where $G_{\xi, \sigma}$ is the generalized Pareto distribution, which for $\xi \neq0$ has a cdf given by
    \begin{equation}
    G_{\xi, \sigma}(y) =
      1-\left (1+\frac{\xi y}{\sigma}\right )^{-1/\xi}.
  \end{equation}
\end{theorem}
Using \cref{PBH}, it is quite straightforward to derive approximate formulas for the VaR and CVaR using the definition of the excess cdf and \cref{var_formula,cvar_formula}, for example, see \citet[Section 7.2.3]{mcneil2005quantitative}. Before stating these formulas, we make precise the choice of function $\sigma(u)$ in \cref{PBH}, which we give next after some needed definitions. Let $U = (1/\bar{F})^{-1}$, the functional inverse of $1/\bar{F}$. Assume such $U$ exists and is twice-differentiable. The following functions will become important tools for characterizing the tail behaviour of $F$.
\begin{definition}
 The \textit{first-} and \textit{second-order auxiliary functions} are defined as, respectively, 
\begin{equation}\label{A_equation}
a(t) = tU^{\prime}(t), \quad A(t) = \frac{tU^{\prime\prime}(t)}{U^\prime(t)} - \xi + 1.
\end{equation}
\end{definition}
For the remainder of this paper, let $\sigma(u) = a(1/\bar{F}(u))$. It is proven in \citet[Corollary 1]{raoult_worms_2003}, with different notation, that \cref{eq:PBH} achieves the optimal rate of convergence with $\sigma(u) = a(1/\bar{F}(u))$ when the following condition on $A$ holds, which we assume to be true for the rest of this paper.
\begin{assumption}\label{A_conditions}
If $F \in \text{MDA}(H_\xi)$, the second-order auxiliary function $A$ exists and satisfies the following conditions: 
\begin{enumerate}[(i)]
\item $\lim_{t \rightarrow \infty} A(t)=0$;
\item $A$ is of constant sign in a neighborhood of $\infty$;
\item $\exists \rho \leq 0$ such that $|A| \in R V_{\rho}$.
\end{enumerate}
\end{assumption}
While \cref{A_conditions} may seem restrictive at first glance, it is in fact a very general condition, satisfied by all common distributions that belong to a maximum domain of attraction \citep{mle_holger}. Counterexamples are fairly contrived and rarely seen in practice, e.g., \citet[Exercise 2.7 on p. 61]{dehaan2006}. 

Now, with a precise definition of $\sigma(u)$, we state the approximations for the VaR and CVaR which follow from \cref{PBH}. For the rest of this paper, we shall denote $s_{u,\alpha} = \bar{F}(u)/(1-\alpha)$. 
\begin{corollary}[POT approximations] \label{pot_approximations}
Suppose that \cref{F_assumption} and \cref{A_conditions} are satisfied. Fix $u \in \mathbb{R}$ and let $\sigma=a(1/\bar{F}(u))$. Then, due to $F_u \approx G_{\xi,\sigma}$ by \cref{eq:PBH}, the POT approximations for the VaR and CVaR are given by, respectively,
\begin{equation} \label{cvar_formula_evt}
q_{u,\alpha}=u+\frac{\sigma}{\xi}\left(s_{u,\alpha}^\xi-1\right), \qquad c_{u,\alpha}= u+ \frac{\sigma}{1-\xi}\left(1 + \frac{s_{u,\alpha}^\xi-1}{\xi}\right).
\end{equation}
\end{corollary}
The accuracy of the POT approximations depends on how high of a threshold is used. When these approximations are used in statistical estimation, a lower threshold is preferable to make use of as much data as possible, but this can induce a significant bias. To estimate this bias, explicit expressions are required for the approximation error when using \cref{cvar_formula_evt}. In the next section, we derive these expressions.

\section{GPD Approximation Error} \label{se:gpd_approx_error}
When applying the POT approximation for the CVaR, there is a deviation between $c_{u,\alpha}$ and $c_\alpha$ that can be quantified asymptotically. We define this deviation as follows.
\begin{definition} The GPD approximation error (of the CVaR) at level $\alpha$ and threshold $u$ is defined as 
\[
\Egpd \triangleq c_{u,\alpha} - c_\alpha.
\]
\end{definition}
Note that $\Egpd$ is a deterministic quantity. We do not yet consider statistical estimation of any parameters, and this is left for subsequent sections. In this section, the asymptotic behaviour of $\Egpd$ as $u\to\infty$ is derived, which leads to a useful approximation for finite $u$. For the rest of this paper, we shall denote $\tau_u = 1/\bar{F}(u)$.
\begin{theorem}\label{gpd_approximation_error}
Suppose \cref{F_assumption} and \cref{A_conditions} hold. Let $\alpha=\alpha_u = 1-\bar{F}(u)/\beta$, where $\beta>1$ is a constant not depending on $u$. Then,
\[
\frac{\Egpd}{a(\tau_u)A(\tau_u)K_{\xi, \rho}(\beta)} \to 1 \quad \text{as} \quad u \to \infty,
\]
where
\begin{equation} \label{gpd_approx_error_cases}
K_{\xi, \rho}(\beta) =
\begin{cases}
\frac{1}{\rho}\left(\frac{\beta^{\xi}}{\xi(1-\xi)}- \frac{1}{\xi+\rho}\left(\frac{\beta^{\xi+\rho}}{(1-\xi-\rho)} + \frac{\rho}{\xi}\right)\right), \quad \rho<0, \, \xi+\rho \neq 0, \\

\frac{1}{\rho}\left(\frac{\beta^{\xi}}{\xi(1-\xi)} -\log{\beta} +\frac{\xi-1}{\xi} \right), \qquad\qquad\; \rho<0, \, \xi+\rho = 0, \\

\frac{\beta^{\xi}}{\xi(1-\xi)}\left(\frac{1-2\xi}{\xi(1-\xi)} - \log{\beta}\right) + \frac{1}{\xi^2},\qquad\quad \rho=0.
\end{cases}
\end{equation}
\end{theorem}
In practice, we would typically be interested in the CVaR at a fixed value of $\alpha$, so it may appear unsatisfactory that $\alpha\to1$ in \cref{gpd_approximation_error}. However, a useful approximation in the non-asymptotic setting which holds for large $u$ is $\Egpd \approx a(\tau_u)A(\tau_u) K_{\xi, \rho}(s_{u,\alpha})$, which is valid as long $\alpha > F(u)$. In subsequent sections, we derive estimators for all needed quantities to estimate $c_{u,\alpha}$ and $\epsilon_{u,\alpha}$ (and thus $c_{\alpha}$) from data, namely the parameters $\xi, \sigma, \rho$, and function $A$, leading to an asymptotically unbiased estimator of $c_\alpha$.

\section{POT Estimator with MLE Bias Correction} \label{se:pot_mle}
In this section, we discuss the estimation of $c_{u,\alpha}$ using \cref{pot_approximations} and maximum likelihood. One possible way to do so is by first selecting a threshold $u$, and then estimating the GPD parameters using the threshold excesses above $u$. Let $X_{(1,n)}\leq X_{(2,n)}\ldots\leq X_{(n,n)}$ denote the order statistics for a sample of size $n$. Let $u = \thresh$ for some value of $k=k_n < n$. Then, the threshold excesses $Y_i = X_{(n-k+i,n)} - u, i=1,..,k$ are i.i.d. \citep[Section 3.4]{dehaan2006} and approximately distributed by a GPD (\cref{PBH}). Maximum likelihood estimators (MLEs) are obtained by maximizing the approximate log-likelihood function with respect to $\xi$ and $\sigma$,
\begin{align}
    (\hXi, \hSig) = \argmax_{\xi,\sigma}\sum_{i=1}^{k}\log g_{\xi,\sigma}(Y_i),\label{e1}
\end{align}
where $g_{\xi,\sigma}$ is the probability density function (pdf) of the GPD, which for $\xi \neq 0$ is given by 
\[
g_{\xi, \sigma}(y)=
\frac{1}{\sigma}\left(1+\frac{\xi y}{\sigma}\right)^{-1 / \xi-1}.
\]
Based on partial derivatives of the log-pdf with respect to parameters, the resulting maximum likelihood first-order conditions when $\xi>0$ are given by
\begin{equation} \label{mle_equations}
\begin{dcases}
\frac{1}{k} \sum_{i=1}^{k} \log \left(1+\frac{\xi Y_i}{\sigma}\right)=\xi,\\
\frac{1}{k} \sum_{i=1}^{k} \frac{Y_i}{\sigma+\xi Y_i}=\frac{1}{\xi+1}.
\end{dcases}
\end{equation}
A closed-form solution to \cref{mle_equations} does not exist, but the MLEs can be obtained numerically through standard software packages. See, for example, \cite{grimshaw} for an overview of the commonly implemented algorithm.

While the usual asymptotic theory of maximum likelihood does not apply in the approximate GPD model, the following theorem establishes the fact that the MLEs are asymptotically normal with a biased mean as long as the number of threshold excesses is chosen suitably. We will include a correction for the asymptotic bias in an estimator for the CVaR subsequently. The following theorem is given in \citet[Theorem 3.4.2]{dehaan2006}.
\begin{theorem} \label{mle_normality} Suppose that \cref{F_assumption} and \cref{A_conditions} hold. Then for $k=k_n \to \infty$ and $k/n\to 0$ as $n \to \infty$, if $\lim_{n\to\infty} \sqrt{k} A(n/k) = \lambda < \infty$, then the MLEs satisfy
$$\sqrt{k}(\hXi - \xi, \; \hSig/a(n/k)-1) \overset{d}{\to} N(\lambda b_{\xi,\rho}, \mathbf{\Sigma}),$$
where $N$ denotes the normal distribution and
\begin{equation} \label{bias_and_variance}
\begin{split}
b_{\xi, \rho} &= \left(b^{(1)}_{\xi, \rho}, b^{(2)}_{\xi, \rho}\right) = \frac{[\xi+1, \; -\rho]}{(1-\rho)(1+\xi-\rho)},\\\mathbf{\Sigma} &= 
\left[\begin{array}{cc}
{(1+\xi)^{2}} & {-(1+\xi)} \\
{-(1+\xi)} & {1+(1+\xi)^2}
\end{array}\right].
\end{split}
\end{equation}
\end{theorem}
For the remainder of this paper, let $u=u_n = \thresh$. In the assumption of \cref{mle_normality}, it does not seem possible to give conditions to guarantee $\sqrt{k}A(n/k)\to\lambda<\infty$ in full generality, but a common approach when working with heavy-tailed distributions is to assume that they belong to the Hall class \citep{hallclass}, which nests those most often seen in practice, for example, the Burr, Fr\'echet, Student, Cauchy, Pareto, $F$, stable etc. The Hall class satisfies \cref{A_conditions} with $A(t) = ct^{\rho}$ for some constant $c\in\mathbb{R}$, and so to ensure convergence we only require that $k = O(n^{-2\rho/(1-2\rho)})$.

To obtain an asymptotically unbiased estimator of the CVaR, we will first correct the asymptotic bias in \cref{mle_normality} using consistent estimators for $A(n/k)$ and $b_{\xi,\rho}$ (which requires an estimator for $\rho$). We use the consistent estimator $\hRho$ of \cite{fraga2} to estimate $\rho$, and a new estimator for $A(n/k)$ is given in \cref{A_estimator}, which we denote $\hat{A}_n$. We prove $\hat{A}_n$ is consistent, in the sense that $\hat{A}_n/A(n/k) \pto 1$, in \cref{A_consistency}. We provide details of the estimators $\hRho$ and $\hat{A}_n$in \cref{se:second-order}. To obtain a consistent estimator for $b_{\xi,\rho}$, it suffices to plug in any consistent estimators for $\xi$ and $\rho$ into \cref{bias_and_variance}, which follows from the continuous mapping theorem (see, for example, \citet[Theorem 2.3]{vaart_1998}). Since $\hXi \overset{p}{\to} \xi$ by \cref{mle_normality}, we set
\begin{equation} \label{bias_estimator}
\hat{b}_n = (\hat{b}^{(1)}_n, \hat{b}^{(2)}_n) \triangleq \frac{[\hXi+1, \, -\hRho ]}{(1-\hRho)(1+\hXi-\hRho)}
\end{equation}
as an estimator for $b_{\xi,\rho}$, where $\hat{b}_n \overset{p}{\to} b_{\xi,\rho}$. We now give bias-corrected estimates of the GPD parameters, which we define by
\begin{equation}
    \hat{\xi}_n \triangleq \hXi - \hat{A}_n\hat{b}^{(1)}_n, \quad
\hat{\sigma}_n \triangleq \hSig(1-\hat{A}_n\hat{b}^{(2)}_n), \label{e3}
\end{equation}
The following theorem shows that $\hat{\xi}_n$ and $\hat{\sigma}_n$ are asymptotically normal and centered with the same asymptotic variance $\mathbf{\Sigma}$ as in \cref{bias_and_variance}.
\begin{theorem} \label{gpd_estimators_normality} Suppose that the assumptions of \cref{mle_normality} hold. Then
\[
\sqrt{k}(\hat{\xi}_n - \xi, \; \hat{\sigma}_n/a(n/k)-1) \overset{d}{\to} N(0, \mathbf{\Sigma}).
\]
\end{theorem}
Using \cref{gpd_estimators_normality}, a new estimator for $c_{u,\alpha}$ can be constructed from \cref{cvar_formula_evt}, which we then show is asymptotically normal and centered. The only missing requirement is an estimate for $F(u)$, which, with $u=X_{(n-k,n)}$, can be obtained using the empirical distribution function, i.e., $\hat{F}_n(u) = 1-k/n$.
\begin{definition}[POT estimator] \label{def:pot_estimator}Suppose that $(\hat{\xi}_n, \hat{\sigma}_n)$ are obtained from $k$ threshold excesses with $\hat{\xi}_n < 1$. Then, an estimator for $c_{u,\alpha}$ at level $\alpha > 1-k/n$ is
\begin{equation}\label{pot_estimator}
\POT \triangleq \frac{\hat{\sigma}_n}{1-\hat{\xi}_n}\left(1+ \frac{1}{\hat{\xi}_n}\left[\left(\frac{k}{n(1-\alpha)}\right)^{\hat{\xi}_n} - 1 \right]\right) + \thresh.
\end{equation}
\end{definition}
Typically, when the CVaR is estimated using the POT approach in the literature, e.g., \cite{mcneil2005quantitative}, \cref{pot_estimator} is used with $(\hXi, \hSig)$ in place of our estimators $(\hat{\xi}_n, \hat{\sigma}_n)$. Hence, the typical approach introduces two sources of bias with respect to the true CVaR: the bias from the MLEs and the bias from the misspecification of the threshold excesses by the GPD (which can be corrected using the GPD approximation error). We now state the main theorem of this section, in which the asymptotic normality of $\POT$ is derived.
\begin{theorem} \label{thm:pot_normality}
Suppose that the assumptions of \cref{mle_normality} hold. Let $\alpha=\alpha_n=1-(1/\beta)k/n$ where $\beta>1$ is a constant not depending on $n$. Let
\[
d_\beta(x,y) = 
\frac{y}{1-x}\left(1 + \frac{\beta^x - 1}{x}\right).
\]
Then, assuming the asymptotic independence of $(\hat{\xi}_n, \hat{\sigma}_n)$ and the random variable $\sqrt{k}\left((k\tau_u/n)^\xi - 1\right)$,
\begin{equation} \label{pot_normality}
\frac{\sqrt{k}}{a(n/k)}\left(\POT - c_{u,\alpha}\right) \overset{d}{\to} N\left(0, V\right),
\end{equation}
where 
$V=\nabla d_{\beta}(\xi, 1)^\top \mathbf{\Sigma} \nabla d_{\beta}(\xi, 1) +1$
and $\nabla d_{\beta}(\xi, 1)$ denotes the gradient of $d_\beta$ evaluated at $(\xi, 1)$, given at the end of \cref{pot_normality_proof}.
\end{theorem}
\begin{remark}
The assumption that $(\hat{\xi}_n, \hat{\sigma}_n)$ is asymptotically independent of $\sqrt{k}\left((k\tau_u/n)^\xi - 1\right)$
seems justified by the proofs of \citet[Theorem A.1 and Theorem 3.1]{BEIRLANT20092800}, where a similar asymptotic independence is established for a bias-corrected Hill estimator of $\xi$. We leave the confirmation of conditions under which this holds for future work. 
\end{remark}
\begin{remark}
The conditions of \cref{thm:pot_normality} imply that $\alpha \to 1$, however, this is not very restrictive in a practical setting since finite sample approximations will be valid for any fixed choice of $\alpha$ as long as $\alpha > 1-k/n$, since $\beta$ is arbitrary.
\end{remark}
 Using \cref{pot_normality} combined with an estimator for $\epsilon_{u,\alpha}$, we derive an asymptotically unbiased estimator and confidence interval for the CVaR in the next section.

\section{Unbiased POT Estimator} \label{se:confidence_intervals}
In the previous section, we derived the asymptotic normality of the POT estimator with bias corrected parameters, $\POT$. While $\POT$ is asymptotically unbiased with respect to $c_{u,\alpha}$, we still need to include the GPD approximation error to correct the remaining deviation induced by the GPD model. For a confidence level $\alpha$ and $u=\thresh$, using \cref{gpd_approximation_error} we can derive an estimator for the GPD approximation error, given by
\begin{equation} \label{approx_error_estimator}
\hEgpd \triangleq \hat{\sigma}_n\hat{A}_n\hat{K}_n,
\end{equation}
where $\hat{K}_n = K_{\hat{\xi}_n,\hat{\rho}_n}(k/(n(1-\alpha)))$, defined in \cref{gpd_approx_error_cases} with known values replaced by their respective estimators. We can now define the following estimator for the CVaR.
\begin{definition}[Unbiased POT estimator] \label{def:unbiased_pot_estimator}The unbiased POT estimator is an estimator for the CVaR at level $\alpha > 1-k/n$, which is defined for $\hat{\xi}_n < 1$, and is given by
\begin{equation}\label{unbiased_pot}
\UPOT \triangleq \POT - \hEgpd.
\end{equation}
\end{definition}
Note that $\UPOT$ is asymptotically unbiased with respect to $c_\alpha$, a statement which is made precise in the following theorem.
\begin{theorem} \label{thm:confidence_int} Suppose that the assumptions of \cref{thm:pot_normality} hold. Then,
\begin{equation} \label{eq:pot_convergence}
\frac{\sqrt{k}(\UPOT - c_\alpha)}{\hat{\sigma}_n\sqrt{\hat{V}_n}} \overset{d}{\to} N(0,1),
\end{equation}
Where $\hat{V}_n$ denotes a consistent estimator of $V$, which can be obtained by plugging in $\hat{\xi}$ into the expression for $V$ given in \cref{thm:pot_normality}.
\end{theorem}
\begin{corollary}
Based on the above limit, an asymptotic confidence interval with level $1-\delta$ for $c_\alpha$ is
\begin{equation} \label{cvar_confidence_interval}
C_\delta^n = \left(\UPOT -z_{\delta/2}\hat{\sigma}_n\sqrt{\hat{V}_n/k},\; \UPOT + z_{\delta/2}\hat{\sigma}_n\sqrt{\hat{V}_n/k}\right),
\end{equation}
where $z_{\delta/2}$ satisfies $\mathbb{P}(Z > z_{\delta/2}) = \delta/2$ with $Z \sim N(0,1)$. \Cref{cvar_confidence_interval} has asymptotically correct coverage probability, i.e., $\Prob(c_\alpha \in C_\delta^n) \to 1-\delta$ as $n \to \infty$.
\end{corollary}

Our confidence interval enables quantifying the level of uncertainty in $\UPOT$. The correct coverage probability of $C_\delta^n$ is a property that has not been previously possible with other CVaR estimators based on the POT approach, where the error from the GPD approximation and the empirical distribution function is ignored. In the next section we give estimators for the second-order parameters $\rho$ and $A(n/k)$ which are needed to compute $\UPOT$ from data. 
\section{Estimation of Second-order Parameters} \label{se:second-order}
\subsection{Estimation of \texorpdfstring{$\rho$}{rho}}\label{se:second-order-rho}
The parameter $\rho$ controls the rate of convergence in \cref{limit_dist} \citep{jackknife_second_order}. The smaller in magnitude the value of $\rho$, the more bias exists in the largest observations from a sample with respect to the GEVD. Therefore, estimates of $\rho$ can be used to control the bias associated with estimates of $\xi$. In our experiments in the next section, we choose the $\rho$ estimator of \cite{fraga2} combined with the adaptive selection of tuning parameters given in \cite{bias_reduce_rho}. Let
\begin{gather*}
M_n^{(j)}(m) = \frac{1}{m}\sum_{i=1}^m[\log X_{(n-i+1,n)} - \log X_{(n-m,n)}]^j,\\
T_n^{(\tau)}(m) = \frac{(M_n^{(1)}(m))^\tau - (M_n^{(2)}(m)/2)^{\tau/2}}{(M_n^{(2)}(m)/2)^{\tau/2} - (M_n^{(3)}(m)/6)^{\tau/3}}, \quad \tau\in\mathbb{R},
\end{gather*}
with the notation $a^{b\tau}=b\log a$ if $\tau=0$. Then, an estimator for $\rho$ is given by \citet[Equation 2.18]{fraga2},
\begin{equation}\label{rho_estimator}
\hat{\rho}_n = \frac{3 (T_n^{(\tau)}(m) - 1)}{T_n^{(\tau)}(m)-3}.
\end{equation}
The number of upper order statistics chosen to estimate $\rho$ is usually much larger than the choice used to estimate $(\xi,\sigma)$, i.e., $m > k$. It is shown in \cite{fraga2} that $\hat{\rho}_n$ is consistent under certain conditions. Let $A_0$ denote the function satisfying the second-order condition of \citet[Equation 2.3.22]{dehaan2006}. If $m=m_n\to\infty, m/n \to 0$ and $\sqrt{m}A_0(n/m) \to \infty$ as $n\to\infty$, then $\hat{\rho}_n \overset{p}{\to} \rho$. The estimator $\hat{\rho}_n$ has an asymptotic bias, and the reduction of this bias is dependant on the choice  of $m$ as well as the tuning parameter $\tau$. Fortunately, the adaptive algorithm given \citet[Section 4.1]{bias_reduce_rho} provides an effective method of bias correction by choosing $m$ and $\tau$ via the most stable sample path of $\hat{\rho}_n$. Details of the full estimation procedure are given in \cref{ada_rho}.

\subsection{Estimation of \texorpdfstring{$A(n/k)$}{A(n/k)}}\label{se:second-order-A}
Currently, no estimators for $A(n/k)$ exist in the literature (to the best of our knowledge). As part of a secondary contribution of this paper, we derive an estimator for $A(n/k)$ in order to estimate $\UPOT$ from i.i.d. samples. Following the formulation of \cite{nawel}, we can adapt their estimator for $A_0(n/k)$ to non-truncated data. Then, using the relation between $A_0$ and $A$ in \citet[Table 3.1]{dehaan2006}, an estimator for $A(n/k)$ is
\begin{equation}\label{A_estimator}
\hat{A}_n \triangleq \frac{(\hXi+\hRho)(1-\hRho)^2(\hat{M}_n^{(2)} - 2(\hat{M}_n^{(1)})^2)}{2\hXi\hRho \hat{M}_n^{(1)}},
\end{equation}
where we define $\hat{M}_n^{(j)} \triangleq M_n^{(j)}(k)$. The proof that $\hat{A}_n$ is consistent is given in \cref{A_consistency}.

\section{Numerical Experiments} \label{se:num_sims}
In this section, we investigate the finite sample performance of $\UPOT$ (denoted \textbf{UPOT} in this section) compared with the sample average estimator (\cref{sa_cvar}), and POT estimator with no bias correction, i.e., \cref{pot_estimator} with $(\hat{\xi}_n,\hat{\sigma}_n)$ replaced by $(\hXi,\hSig)$. Denote these estimators as  \textbf{SA} and \textbf{BPOT}, respectively. First, in the theoretical setting, we compare the exact values of the asymptotic variance of \textbf{UPOT} and \textbf{SA} at different values of $\alpha$ and sample sizes on the Fr\'echet distribution. This analysis provides justification for the cases where \textbf{UPOT} is expected to perform better than \textbf{SA} on data. Next, we assess the statistical accuracy of the three estimation methods at different sample sizes among several classes of heavy-tailed distributions. Finally, we assess the accuracy of the asymptotic confidence interval given in \cref{cvar_confidence_interval} on finite samples by using the empirical coverage probability.

\subsection{Comparison of Asymptotic Variance} \label{asymp_var_compare}
In this section, the magnitude of the asymptotic variance (AVAR) of \textbf{UPOT} and \textbf{SA} are compared. Since both estimators are asymptotically unbiased and assuming they are both efficient, the mean squared error of each estimator approaches the AVAR in large samples (by the Cram\'er-Rao lower bound). Hence, this comparison gives evidence of the distributional properties and level of $\alpha$ where \textbf{UPOT} results in lower error than \textbf{SA}. The comparison is made on the Fr\'echet distribution with single parameter $\gamma,$ which has $\xi=1/\gamma$, $\rho=-1$ (see \cref{se:frechet}). We compute $V/k$ (given in \cref{thm:pot_normality}) with $n=10000,20000,\ldots,100000$ and set $k=\lceil n^{2/3}\rceil$ to satisfy the assumption of \cref{mle_normality}. An expression for the AVAR of \textbf{SA} is given in, for example, \cite{TRINDADE20073524}, and we provide the details of this calculation for the Fr\'echet distribution in \cref{asymp_sa_frec}. To the best of our knowledge, the AVAR of \textbf{SA} can only be derived for distributions with a bounded second moment, which corresponds to distributions with $\xi < 1/2$ (or $\gamma>2$ in the Fr\'echet case). The AVAR of \textbf{SA} and \textbf{UPOT} is compared for the Fr\'echet distribution with $\gamma=2.25, 2.5, 3, 4$ and $\alpha=0.99, 0.999$ in \cref{fig:asympvar}. The results indicate that \textbf{UPOT} is preferable for high values of $\alpha$ and low values of $\gamma$. Increasing $\alpha$ would lead to lower sample availability in \textbf{SA}, and thus higher variance, while \textbf{UPOT} is unaffected. Decreasing $\gamma$ is equivalent to increasing $\xi$ and thus increasing tail thickness. This increases the AVAR of \textbf{SA} since extreme observations are much further from the mean but not readily observed. Based on evidence from the Fr\'echet distribution, it is reasonable to extrapolate that \textbf{UPOT} should always perform better than \textbf{SA} on heavy-tailed distributions with $\xi\ge 1/2$ at high values of $\alpha$.

\begin{figure*}[ht!]
    \centering
        \includegraphics[width=\textwidth]{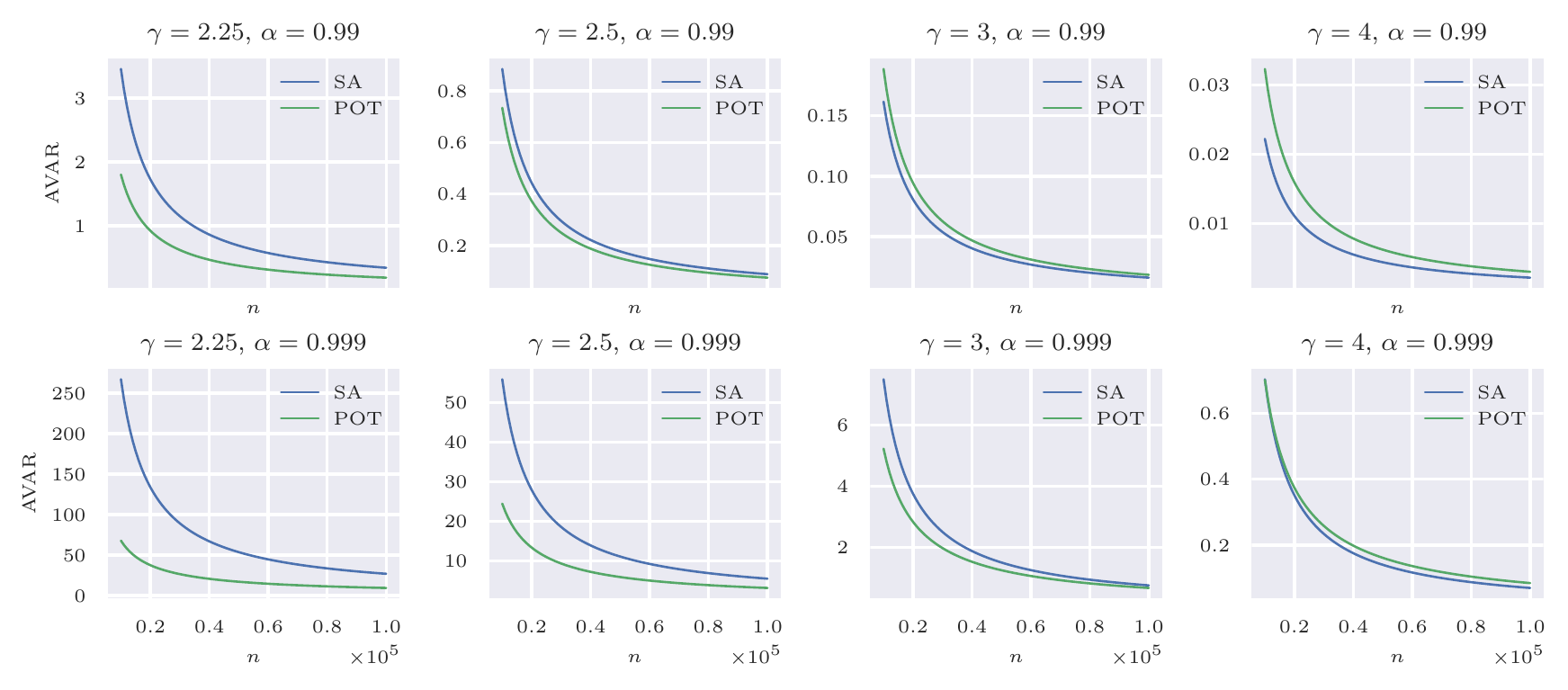}
    \caption{Asymptotic variance of the SA CVaR estimator (blue) and unbiased POT estimator (green) for the Fr\'echet($\gamma$) distribution at $\alpha=0.99,\,0.999$.}
\label{fig:asympvar}
\end{figure*}

\subsection{Error Analysis of CVaR Estimators} \label{se:error_analysis}
In the experiments that follow, samples are generated from the Burr, Fr\'echet, and half-$t$ distributions, which provide a good characterization of heavy-tailed phenomena with finite mean. Relevant details for each distribution class are provided in \cref{appendix:dists}. The estimation performance of \textbf{SA}, \textbf{BPOT}, and \textbf{UPOT} are compared via the root-mean-square error (RMSE) and absolute bias on five examples from each distribution class, shown in \cref{fig:compare}. We fix $\alpha=0.998$ as an example of an extreme risk level. Experiments are conducted as follows. Generate $N=1000$ random samples of size $50000$ from each distribution. For each sample, the CVaR is estimated using the three methods at subsample sizes $n=5000,10000,\ldots,50000$. In practice, it can be difficult to choose the number of threshold excesses $k$, and so we apply the ordered goodness-of-fits tests of \cite{bader2018automated} to
choose the optimal threshold. This threshold selection procedure, which we employ in both \textbf{BPOT} and \textbf{UPOT}, is given in detail in \cref{auto_thresh}. The average threshold selected (in terms of the percentile of a given sample) was between 0.80 and 0.96 in all simulations performed. The complete algorithm for \textbf{UPOT} is summarized in \cref{upot_algo}. 

\textbf{Discussion.} The chosen Burr distributions allow us to investigate the effect of varying $\rho$ while keeping a fixed $\xi$. In this case, we set $\xi=2/3$ while $\rho=-0.25, -0.33, -0.44, -1.33, -2.22$ in the respective Burr distributions. In general, when $\rho$ approaches $0$, the distribution's tail deviates more severely from a strict Pareto model, and therefore we see the largest bias and RMSE occur in \textbf{BPOT} in the Burr(0.38, 4) and Burr(0.5, 3) models, while the bias-correction of \textbf{UPOT} leads to the most substantial performance gain. As a non-parametric estimator, \textbf{SA} is less affected by changes in the value of $\rho$, outperforming the POT estimators in terms of bias on some Burr distributions. However, as alluded to in \cref{asymp_var_compare}, high values of $\xi$ leads to high variance in observations, typically causing poor performance of \textbf{SA} in terms of RMSE. This effect is similarly observed in the Fr\'echet simulations, where \textbf{SA} has relatively low bias. The Fr\'echet distribution always has $\rho=-1$, a property shared with the GPD, giving its tail a similar shape. Therefore, the bias-correction of \textbf{UPOT} is less significant, but still provides a noticeable performance gain over \textbf{BPOT}. The results of the half-$t$ simulations are similar to the Fr\'echet, but we note a larger bias in \textbf{BPOT} due to the fact that the half-$t$ distribution has a $\rho$ value that varies with its parameter. Like in the Burr simulations, \textbf{SA} is unaffected by different values of $\rho$ and obtains good performance in terms of bias in the half-$t$ simulations, except when $\xi$ is largest in the half-$t$(1.5) model. Finally, we note that \textbf{UPOT} consistently had the lowest RMSE in all simulations except in a few cases at a sample size of $5000$. Next, the finite sample performance of the \textbf{UPOT} confidence interval is investigated.

\begin{figure*}[ht!]
    \centering
        \includegraphics[width=\textwidth]{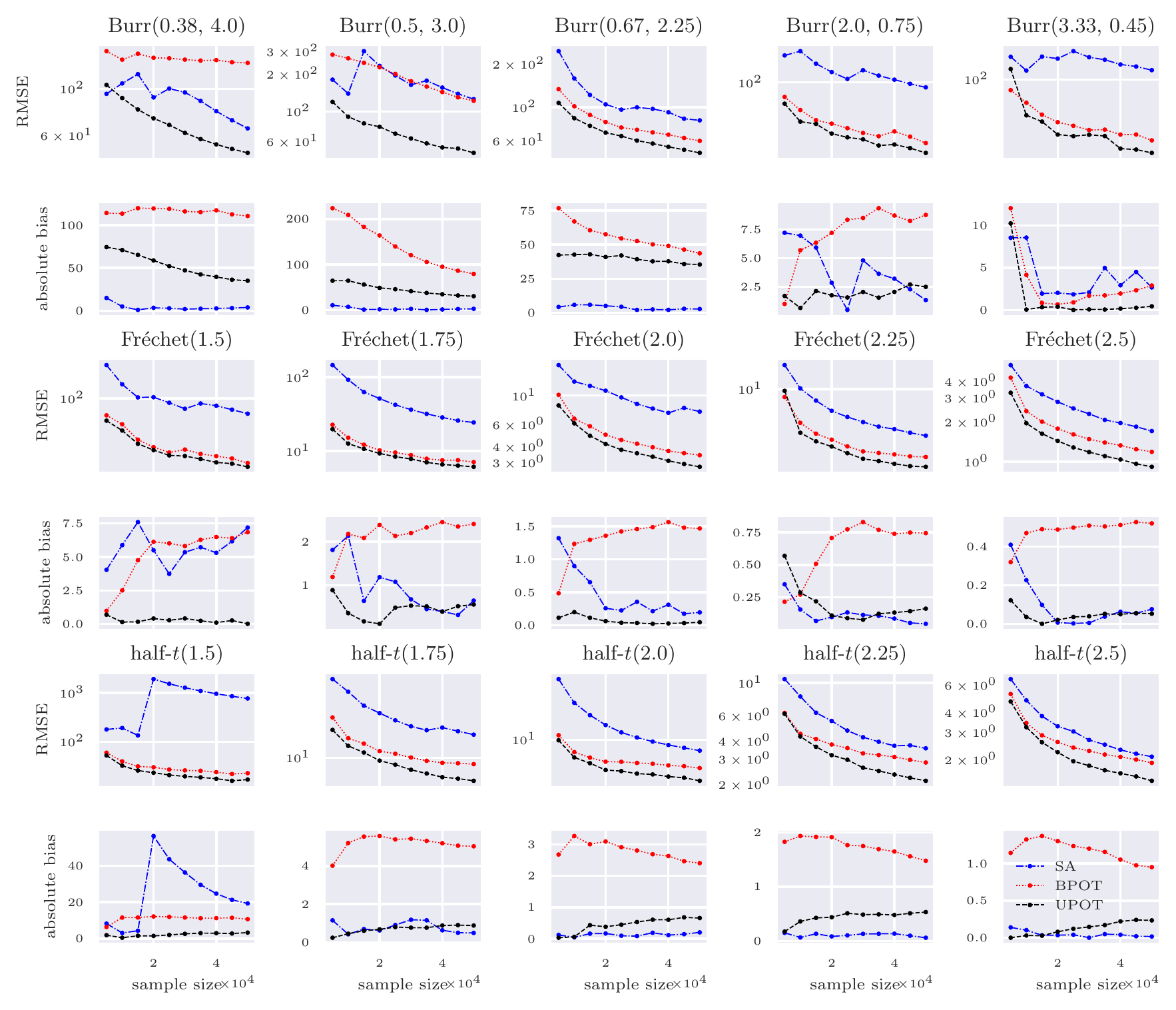}
    \caption{RMSE of and absolute bias of estimating CVaR$_{0.998}$ using \textbf{UPOT} (black), \textbf{BPOT} (red), and \textbf{SA} (blue).}
\label{fig:compare}
\end{figure*}

\subsection{Coverage Probability of the Asymptotic Confidence Interval}
The accuracy of the confidence interval given in \cref{cvar_confidence_interval} is assessed by its empirical coverage probability for each distribution using the same simulated data from \cref{se:error_analysis}. Let $C^n_{i,\delta}$ denote the confidence interval computed for a sample of size $n$ for sample $i$, $i=1,\ldots,N$. Then, the empirical coverage probability is defined as
\[
\hat{P}_\delta^n(N) = \frac{1}{N}\sum_{i=1}^N \mathds{1}_{\left\{c_\alpha \in C^n_{i,\delta}\right\}}.
\]
Plots of the coverage probability at each sample size for each distribution are shown in \cref{fig:coverage}. We set $\delta=0.05$ and compute the coverage probability at sample sizes $n=5000,10000,\ldots,50000$. The final value of each distribution's coverage probability at $n=50000$ is reported in \cref{se:num_results}. Most of the distributions tested achieve nearly the correct coverage of 0.95, sometimes surpassing it in some cases, and this is due to the estimated confidence interval being wider than its true asymptotic counterpart. The coverage is worst in the Burr(0.38, 4) distribution, achieving a final coverage probability of just 0.73. The small magnitude of $\rho$ in this distribution causes slow convergence of the tail to the GPD, and hence a relatively high average threshold percentile of 0.96 was chosen by the threshold selection procedure. This high threshold increases the variance of parameter estimation which explains the poor coverage. 
\begin{figure*}[ht!]
    \centering
        \includegraphics[width=\textwidth]{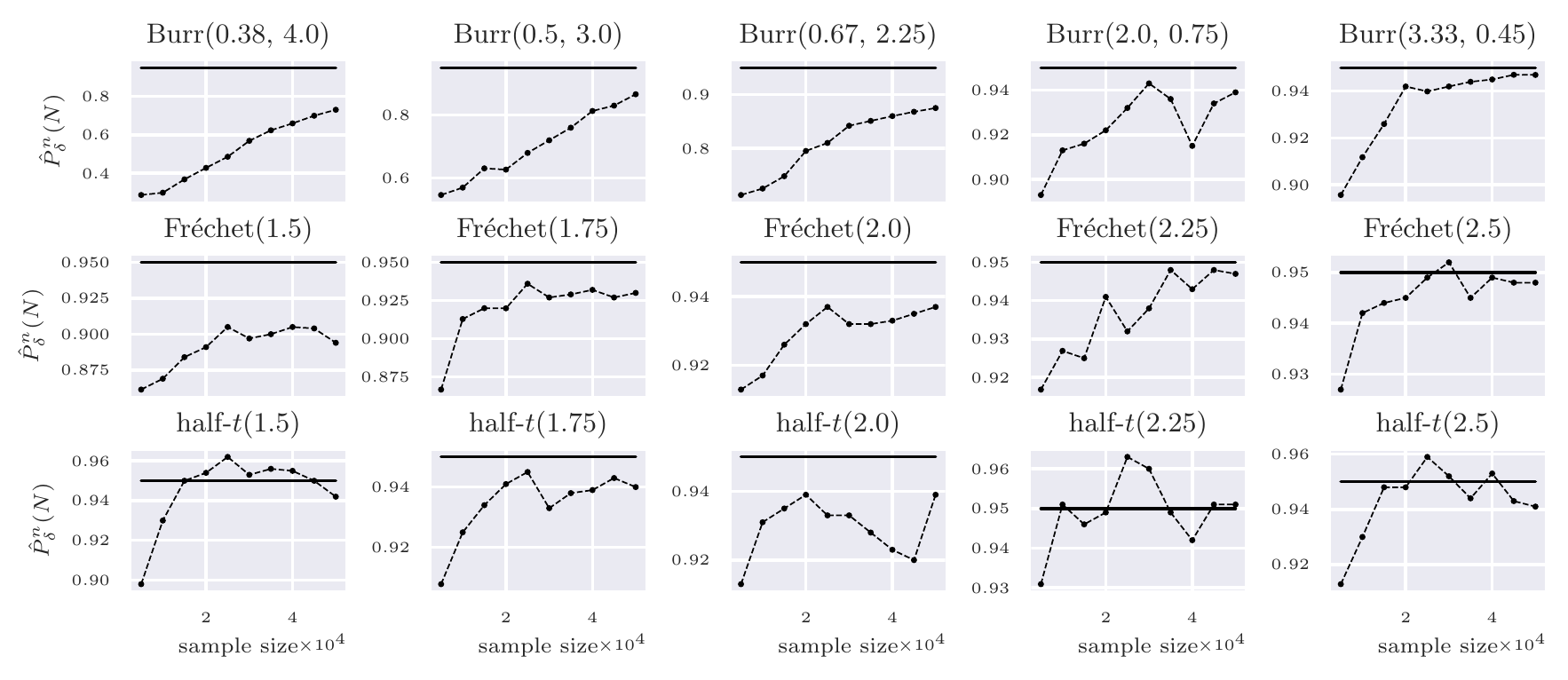}
    \caption{Coverage probabilities with $\alpha=0.998, \delta = 0.05$. The solid line indicates the theoretical coverage, i.e., $1-\delta = 0.95$.}
\label{fig:coverage}
\end{figure*}

\section{Conclusion} \label{se:conclusion}
We have studied the asymptotic properties of a new CVaR estimator based on the peaks-over-threshold approach. Using extreme value theory and second-order regular variation, we derived estimators for the bias induced by the approximate GPD model of the threshold excesses and the bias from maximum likelihood estimators of the GPD parameters. Using these results, we proved that our estimator is asymptotically normal and unbiased (up to some technical conditions). This convergence result allowed us to derive confidence intervals for the CVaR, enabling us to measure the level of uncertainty in our estimator. We compared the magnitudes of the asymptotic variance of our CVaR estimator with that of the sample average CVaR estimator, demonstrating a significant improvement in asymptotic performance for some cases. An empirical study showed that our CVaR estimator can lead to a significant performance improvement in heavy-tailed distributions when compared to the sample average estimator and the existing peaks-over-threshold estimator. Finally, we investigated the finite-sample performance of the asymptotic confidence interval, and found that good coverage probability is achieved in reasonable sample sizes. While our evidence suggests that our CVaR estimator is most effective in the heavy-tailed domain, it would also be instructive to perform the same theoretical analysis for light-tailed distributions. Doing so would allow our CVaR estimator to be robust to situations where it is not possible to make any assumptions about the underlying data distribution.

\bibliographystyle{apalike}
\bibliography{Bib}

\appendix
\section{Proofs} \label{appendix:proofs}
We first recall the stochastic order notation (e.g., \citet[Section 2.2]{vaart_1998}), which will be used throughout subsequent proofs.

\begin{definition}[Stochastic $o$ and $O$ symbols]
Let $X_n, R_n$ denote sequences of random variables. Then,
\begin{align*}
X_n &= o_p(R_n) \quad \textnormal{ means } \quad \forall\varepsilon>0, \, \lim_{n\to\infty} \Prob(|X_n/R_n| > \varepsilon) = 0, \\
X_n &= O_p(R_n) \quad \textnormal{ means } \quad \forall\varepsilon>0, \, \exists M,N > 0, \, \forall n > N, \, \Prob(|X_n/R_n| > M) < \varepsilon.
\end{align*}
The often used notation $X_n=o_p(1)$ means that $X_n$ converges to zero in probability, and $X_n=O_p(1)$ means that $X_n$ is bounded in probability.
\end{definition}
\subsection{Proof of Theorem 3.1}
We first state the following lemma, which is equivalent to \citet[Proposition 1]{Beirlant2003OnTR} with different notation.
\begin{lemma} \label{Ilimbound} Suppose assumption 2.1 and assumption 2.2 hold. Then 
$\forall \varepsilon>0, \; \exists t_0, \; \forall t, x$ such that $t \ge t_0$ and $tx \ge t_0$,
\begin{equation} \label{Iboundineq}
(1-\varepsilon) e^{-\varepsilon|\log x|} \leq\left[\frac{U(tx)-U(t)}{a(t)}-\frac{x^{\xi}-1}{\xi}\right] /\left[A\left({t}\right) I_{\xi, \rho}(x)\right] \leq(1+\varepsilon) e^{\varepsilon|\log x|},
\end{equation}
where
\[
A(t) = \frac{tU^{\prime\prime}(t)}{U^\prime(t)} - \xi + 1 \quad \textnormal{ and } \quad 
I_{\xi,\rho}(x) = 
\begin{cases}
\frac{1}{\rho}\left(\frac{x^{\xi+\rho}-1}{\xi+\rho} - \frac{x^{\xi}-1}{\xi}\right), \quad \rho<0, \, \xi+\rho \neq 0, \\
\frac{1}{\rho}\left(\log{x} - \frac{x^{\xi}-1}{\xi}\right), \quad\;\;\;\, \rho<0, \, \xi+\rho = 0, \\
\frac{1}{\xi}\left(x^{\xi}\log{x} - \frac{x^{\xi}-1}{\xi}\right), \;\;\; \rho=0.
\end{cases}
\]
\end{lemma}
\begin{proof}
In \cite{Beirlant2003OnTR}, the statement is given as $\forall \varepsilon>0, \; \exists t_0, \; \forall t, x$ such that $t \ge t_0$ and $t+x \ge t_0$,
\begin{equation}\label{beirlant}
(1-\varepsilon) e^{-\varepsilon|x|} \leq\left[\frac{V(t+x)-V(t)}{V^{\prime}(t)}-\frac{e^{\xi x}-1}{\xi}\right] /\left[\tilde{A}\left(e^{t}\right) \tilde{I}_{\xi, \rho}(x)\right] \leq(1+\varepsilon) e^{\varepsilon|x|},
\end{equation}
where 
\[
V(t)=(\bar{F})^{-1}\left(e^{-t}\right), \qquad \tilde{A}(t)=\frac{V^{\prime \prime}(\log t)}{V^{\prime}(\log t)}-\xi, \qquad \tilde{I}_{\xi, \rho}(x) = I_{\xi, \rho}(e^x).
\]
Then, for $t \ge 1$,
\[
V(\log t) = (\bar{F})^{-1}(1/t) = (1/\bar{F})^{-1}(t) = U(t),
\]
and
\[
V^\prime(\log t) = tU^\prime(t) = a(t), \quad V^{\prime\prime}(\log t) = t^2 U^{\prime\prime}(t) + tU^{\prime}(t) \quad \Rightarrow \quad \tilde{A}(t) = A(t).
\]
Since $\log$ is strictly increasing, \cref{beirlant} holds with $\log t$ and $\log x$ where $\log t \ge t_0$ and $\log tx \ge t_0$. Substituting expressions in \cref{beirlant}, we get \cref{Iboundineq}.
\end{proof}
The following corollary will also be used in the main proof of this section.
\begin{corollary} \label{I_limit_cor} An immediate consequence of \cref{Ilimbound} is for all $x>0,$
\begin{equation}\label{I_limit}
\lim _{t \rightarrow \infty} \frac{\frac{U(t x)-U(t)}{a(t)}-\frac{x^{\xi}-1}{\xi}}{A(t)} = I_{\xi, \rho}(x).
\end{equation}
\end{corollary}
\cref{I_limit_cor} can also be found in \citet[Theorem 2.3.12]{dehaan2006}.
Before proving our main result of this section, we first recall the dominated convergence theorem which will be needed later.
\begin{theorem}[Dominated convergence theorem]\label{thm:dom}
Let $\{f_n\}_{n=1}^\infty$ be a sequence of real-valued functions defined on $S \subset \mathbb{R}$ such that $\forall x \in S, \, \lim_{n\to\infty} f_n(x) \to f(x)$. If $\forall x \in S, \, n,$
\[
|f_n(x)| \le g(x)
\]
for some integrable (i.e., the integral is finite over $S$) function $g$, then
\[
\lim_{n\to\infty} \int_S f_n(x)dx = \int_S \lim_{n\to\infty} f_n(x)dx = \int_S f(x)dx.
\]
\end{theorem}

\vspace{0.5cm}

\textbf{Proof of Theorem 3.1.} We use \cref{I_limit_cor} to derive a convergence result for the approximation error of the VaR, i.e, $q_\alpha - q_{u,\alpha}$. Then, using \cref{Ilimbound} and \cref{thm:dom}, we will be able to derive the convergence of $\epsilon_{u,\alpha}$.

For any $p \in (0,1)$ and $y \in \dom F$ such that $F(y) = p$,
\[\left(\frac{1}{\bar{F}}\right)(y) = \frac{1}{\bar{F}(y)} = \frac{1}{1-p},
\]
which implies that
\[U\left(\frac{1}{1-p}\right) = \left(\frac{1}{\bar{F}}\right)^{-1}(1/(1-p)) = y.
\]
Hence,
\[U(1/(1-\alpha)) = U(\tau_u\beta) = q_\alpha \quad \textnormal{ and } \quad U(\tau_u) = u.
\]
Then, from the definition of $q_{u,\alpha}$ we get
\[
q_{u,\alpha} = u+\frac{\sigma(u)}{\xi}\left(\beta^\xi - 1\right) = U(\tau_u) + \frac{a(\tau_u)}{\xi}(\beta^\xi - 1).
\]
Setting $D_u(\beta)=(q_{\alpha}-q_{u,\alpha})/(a(\tau_u) A(\tau_u))$, it then follows from the previous two equations and \cref{I_limit_cor} with $t=\tau_u$, $x=\beta$ that 
\begin{equation}\label{q_limit}
D_u(\beta) = \frac{U(\tau_u\beta) - U(\tau_u) - \frac{a(\tau_u)}{\xi}(\beta^\xi - 1)}{a(\tau_u)A(\tau_u)} =  \frac{\frac{U(\tau_u\beta) - U(\tau_u)}{a(\tau_u)} - \frac{\beta^\xi-1}{\xi}}{A(\tau_u)}  \to I_{\xi,\rho}(\beta) \quad \textnormal{ as } u \to \infty.
\end{equation}
From the definition of the GPD approximation error and the CVaR, for a fixed $u,\alpha$,
\begin{equation}\label{gpd_ae_integral}
\frac{\epsilon_{u,\alpha}}{a(\tau_u)A(\tau_u)} = \frac{c_{u,\alpha} - c_{\alpha}}{a(\tau_u)A(\tau_u)} = -\frac{1}{1-\alpha}\int_\alpha^1 \frac{q_{\gamma}-q_{u,\gamma}}{a(\tau_u)A(\tau_u)} d\gamma = -\beta\int_\beta^\infty \frac{D_u(x)}{x^2} dx,
\end{equation}
where $D_u(x)=(q_{\gamma}-q_{u,\gamma})/(a(\tau_u) A(\tau_u))$ and we have used the substitution $x=\bar{F}(u)/(1-\gamma)$.
We now apply the dominated convergence theorem to get the limiting behaviour of \cref{gpd_ae_integral} as $u\to\infty$. From \cref{Ilimbound}, $\forall \varepsilon > 0, \, \exists u_0$ such that $\forall u \geq u_0, \, x \in [\beta, \infty)$,
\[
\left|\frac{D_u(x)}{x^2}\right| \le (1+\varepsilon) x^{\varepsilon-2}I_{\xi,\rho}(x).
\]
$(1+\varepsilon) x^{\varepsilon-2}I_{\xi,\rho}(x)$ is integrable over $[\beta, \infty)$ as long as $\varepsilon < 1-\xi$. Since $\xi<1$, let $\varepsilon=(1-\xi)/2$. Then \cref{thm:dom} can be applied to $D_u(x)/x^2$. Setting $K_{\xi,\rho}(\beta) = -\beta\int_\beta^\infty [I_{\xi,\rho}(x)/x^2]dx$, it follows that
\begin{align*}
\lim_{u\to\infty}\frac{\epsilon_{u,\alpha}}{a(\tau_u)A(\tau_u)} &=\lim_{u\to\infty}-\beta\int_\beta^\infty \frac{D_u(x)}{x^2} dx\\ &= -\beta\int_\beta^\infty \lim_{u\to\infty}\frac{D_u(x)}{x^2} dx\\ &= -\beta\int_\beta^\infty \frac{I_{\xi,\rho}(x)}{x^2} dx\\ &= K_{\xi,\rho}(\beta),
\end{align*}
where the last integral can be computed explicitly to obtain \cref{gpd_approx_error_cases}.
\qed

\subsection{Proof of Theorem 4.2}
First recall Slutsky's lemma (see, for example, \citet[Lemma 2.8]{vaart_1998}).
\begin{lemma}[Slutsky]\label{slutsky}
Let $X_n, \, X, \, Y_n$ be random vectors or variables. If $X_n\dto X$ and $Y_n\pto c$ for a constant $c$, then
\begin{enumerate}[(i)]
    \item $X_{n}+Y_{n} \dto X+c$;
    \item $X_{n} Y_{n} \dto X c$;
    \item $X_{n} / Y_{n} \dto X / c$ provided $c \neq 0.$
\end{enumerate}
\end{lemma}

\vspace{0.5cm}

\textbf{Proof of Theorem 4.2.}
First note that since $\hat{A}_n$ and $\hat{b}_n$ are consistent, i.e.,
\[
\frac{\hat{A}_n}{A(n/k)} \pto 1, \qquad \hat{b}_n \pto b_{\xi,\rho}, 
\]
then by \cref{slutsky}, the fact that $\hSig/a(n/k) \pto 1$ (which follows from \cref{mle_normality}), and the assumption of \cref{mle_normality} that $\lim_{n\to\infty}\sqrt{k}A(n/k) \to \lambda < \infty$,
\begin{align*}
\sqrt{k}\hat{A}_n\left(\hat{b}^{(1)}_n, \;\frac{\hSig}{a(n/k)}\hat{b}^{(2)}_n\right) &= \sqrt{k}A(n/k)\frac{\hat{A}_n}{A(n/k)}\left(\hat{b}^{(1)}_n, \;\frac{\hSig}{a(n/k)}\hat{b}^{(2)}_n\right)\\
&\pto \lambda \left(b_{\xi,\rho}^{(1)}, \; b_{\xi,\rho}^{(2)}\right) \\&= \lambda b_{\xi,\rho}.
\end{align*}
Then, by expanding terms and applying \cref{slutsky} once again,
\begin{align*}
\sqrt{k}(\hat{\xi}_n - \xi, \; \hat{\sigma}_n/a(n/k)-1) &= \sqrt{k}(\hXi - \hat{A}_n\hat{b}^{(1)}_n - \xi, \;  \hSig(1-\hat{A}_n\hat{b}^{(2)}_n)/a(n/k) - 1) \\
&= \sqrt{k}(\hXi - \xi, \; \hSig/a(n/k)-1) - \sqrt{k}\hat{A}_n\left(\hat{b}^{(1)}_n, \;\frac{\hSig}{a(n/k)}\hat{b}^{(2)}_n\right) \\
&\dto N(\lambda b_{\xi,\rho}, \mathbf{\Sigma}) - \lambda b_{\xi,\rho} \\
&= N(0, \mathbf{\Sigma}).
\end{align*}
\qed
\subsection{Proof of Theorem 4.3}\label{pot_normality_proof}
We first give the delta method, which can be found in, for example, \citet[Appendix B.3.4.1]{remillard2016statistical}.
\begin{theorem}[Delta method] \label{delta_method}
Let $\hat{\theta}_n \in \mathbb{R}^m$ be a random vector based on a sample of size $n$. Suppose that $h:\mathbb{R}^m \mapsto \mathbb{R}$ is such that for $i=1,\ldots,m$, $\frac{\partial h}{\partial \theta_i}$ exists and is continuous in a neighborhood of $\theta$. If $\sqrt{n}(\hat{\theta}_n - \theta) \dto N(0, V)$, then
$$\sqrt{n}(h(\hat{\theta}_n) - h(\theta)) \dto N(0, \nabla h(\theta)^\top V \nabla h(\theta)),$$
where $\nabla h(\theta)$ is the gradient of $h$ evaluated at $\theta$.
\end{theorem}

Next, we prove some useful lemmas which will be used in the proof of \cref{thm:pot_normality}.
\begin{lemma} \label{tau_lemma}
Let $X_1,\ldots,X_n$ be an i.i.d. sample with common cdf $F$, and suppose $k=k_n\to\infty$ and $k/n\to 0$ as $n\to\infty$. With $u=\thresh$ and $\xi \in \mathbb{R}$,
\[
\sqrt{k}\left((k\tau_u/n)^\xi - 1\right) \dto N(0, \xi^2).
\]
\begin{proof}
Letting $h_\xi(x) = x^{-\xi},$
\[
\sqrt{k}\left((k\tau_u/n)^\xi - 1\right) = \sqrt{k}\left((n\bar{F}(u)/k)^{-\xi} - 1\right) = \sqrt{k}\left(h_\xi(n\bar{F}(u)/k) - h_\xi(1)\right).
\]
From \citet[Theorem 3.1]{BEIRLANT20092800}, we know that $\sqrt{k} (n\bar{F}(u)/k - 1) \dto N(0,1).$ Hence, by \cref{delta_method},
\[
\sqrt{k}\left((k\tau_u/n)^\xi - 1\right) \dto N(0, h_\xi^\prime(1) \cdot 1 \cdot h_\xi^\prime(1)) = N(0, \xi^2).
\]
\end{proof}
\end{lemma}
\begin{corollary} \label{su_corr}
Let $\alpha=\alpha_n=1-(1/\beta)k/n$ where $\beta>1$ is a constant not depending on $n$. Then,
\[
\sqrt{k}(s_{u,\alpha}^\xi - \beta^\xi) \dto N(0, \xi^2\beta^{2\xi}).
\]
\begin{proof}
\[
s_{u,\alpha} = \bar{F}(u)\left(\frac{n}{k}\right)\frac{k}{n(1-\alpha)} = \beta n\bar{F}(u)/k = \frac{\beta n}{k\tau_u},
\]
and so
\begin{equation}\label{su_beta_form}
\sqrt{k}(s_{u,\alpha}^\xi - \beta^\xi) = \beta^\xi\sqrt{k} ((k\tau_u /n)^{-\xi} - 1) = -\beta^\xi (n\bar{F}(u)/k)^\xi \sqrt{k} ((k\tau_u /n)^{\xi} - 1).
\end{equation}
\citet[Theorem 3.1]{BEIRLANT20092800} implies that $n\bar{F}(u)/k \pto 1$. Hence, \cref{tau_lemma} with \cref{su_beta_form} implies
\[
\sqrt{k}(s_{u,\alpha}^\xi - \beta^\xi) \dto N(0, \xi^2\beta^{2\xi}).
\]
\end{proof}
\end{corollary}

\begin{lemma} \label{ank_lemma}
Suppose that the assumptions of \cref{mle_normality} hold. Then as $n\to\infty$,
\begin{equation} \label{eq:ank_lemma}
\sqrt{k}\left(\frac{a(\tau_u)}{a(n/k)} - (k\tau_u/n)^\xi\right) = o_p(1).
\end{equation}
\begin{proof}
Under \cref{F_assumption} and \cref{A_conditions}, the following uniform inequality from \citet[Theorem 2.3.6]{dehaan2006} holds: for any $\varepsilon, \delta > 0$ there exists $t_0 = t_0(\varepsilon, \delta)$ such that for all $t, tx \ge t_0$,
\begin{equation} \label{second_order_lemma}
\left|\frac{\frac{a(t x)}{a(t)}-x^{\xi}}{A(t)}-x^{\xi} \frac{x^{\rho}-1}{\rho}\right| \leq \varepsilon x^{\xi+\rho} \max \left(x^{\delta}, x^{-\delta}\right).
\end{equation}
Hence, with $t=n/k$ and $x=k\tau_u/n$, for any $\varepsilon, \delta > 0$ and with large enough $n$,
\begin{align*}
\sqrt{k}\left(\frac{a(\tau_u)}{a(n/k)} - (k\tau_u/n)^\xi\right) &= \sqrt{k}A(n/k) \left[\frac{\frac{a(\tau_u)}{a(n/k)} - (k\tau_u/n)^\xi}{A(n/k)} - (k\tau_u/n)^\xi\frac{(k\tau_u/n)^\rho - 1}{\rho}\right] \\&\qquad+ \sqrt{k}A(n/k)(k\tau_u/n)^\xi\frac{(k\tau_u/n)^\rho - 1}{\rho}\\ 
&\le \sqrt{k}A(n/k) \left\vert\frac{\frac{a(\tau_u)}{a(n/k)} - (k\tau_u/n)^\xi}{A(n/k)} - (k\tau_u/n)^\xi\frac{(k\tau_u/n)^\rho - 1}{\rho}\right\vert \\&\qquad+ \sqrt{k}A(n/k)(k\tau_u/n)^\xi\frac{(k\tau_u/n)^\rho - 1}{\rho} \\
&\leq \sqrt{k}A(n/k)\varepsilon (k\tau_u/n)^{\xi+\rho} \max \left((k\tau_u/n)^{\delta}, (k\tau_u/n)^{-\delta}\right) \\
&\qquad+\sqrt{k}A(n/k)(k\tau_u/n)^\xi\frac{(k\tau_u/n)^\rho - 1}{\rho}.
\end{align*}
Since $k\tau_u/n \pto 1$ and $\sqrt{k}A(n/k) \to \lambda < \infty$ as $n\to\infty$ (by the assumption of \cref{mle_normality}), and since $\varepsilon$ can be made arbitrarily small as $n\to\infty$, both terms tend to $0$ in probability as $n\to\infty,$ hence \cref{eq:ank_lemma} follows.
\end{proof}
\end{lemma}
The following corollary is an immediate result by combining \cref{tau_lemma} and \cref{ank_lemma}.
\begin{corollary} \label{ank_corr}
Suppose that the assumptions of \cref{mle_normality} hold. Then,
\[
\sqrt{k}\left(\frac{a(\tau_u)}{a(n/k)} - 1\right) \dto N(0, \xi^2).
\]
\end{corollary}

\vspace{0.5cm}

\textbf{Proof of Theorem 4.3.}
With $\alpha=1-(1/\beta)k/n$ and $u=\thresh$,
\[
\frac{\hat{c}^{(n)}_{\alpha}}{a(n/k)} = 
\frac{\hat{\sigma}_n/a(n/k)}{1-\hat{\xi}_n}\left(1+ \frac{\beta^{\hat{\xi}_n} - 1}{\hat{\xi}_n}\right) + \frac{u}{a(n/k)} = d_\beta(\hat{\xi}_n, \hat{\sigma}_n/a(n/k)) + \frac{u}{a(n/k)},
\]
and recalling that $s_{u,\alpha} = \bar{F}(u)/(1-\alpha)$,
\begin{align*}
\frac{c_{u,\alpha}}{a(n/k)} &=
\frac{\sigma(u)/a(n/k)}{1-{\xi}}\left(1+ \frac{s_{u,\alpha}^{{\xi}} - 1}{{\xi}}\right) + \frac{u}{a(n/k)} \\
&= \frac{1}{1-{\xi}}\left(1+ \frac{s_{u,\alpha}^{{\xi}} - 1}{{\xi}}\right)  + \frac{\sigma(u)/a(n/k) - 1}{1-{\xi}}\left(1+ \frac{s_{u,\alpha}^{{\xi}} - 1}{{\xi}}\right) + \frac{u}{a(n/k)}.
\end{align*}
Then for the first term,
\begin{align*}
\frac{1}{1-{\xi}}\left(1+ \frac{s_{u,\alpha}^{{\xi}} - 1}{{\xi}}\right)
&= \frac{1}{1-{\xi}}\left(1+ \frac{s_{u,\alpha}^{{\xi}} - 1 +\beta^\xi -\beta^\xi}{{\xi}}\right) \\
&= \frac{1}{1-{\xi}}\left(1+ \frac{\beta^\xi-1}{{\xi}} +\frac{s_{u,\alpha}^{{\xi}} -\beta^\xi}{{\xi}}\right) \\
&= d_\beta(\xi, 1) + \frac{s_{u,\alpha}^{{\xi}} -\beta^\xi}{\xi (1-\xi)}.
\end{align*}
 Recalling that $\sigma(u)=a(\tau_u)$ and combining the previous expressions,
\begin{multline} \label{delta_plus_remainder}
\frac{\sqrt{k}}{a(n/k)}\left(\hat{c}^{(n)}_{\alpha} - c_{u,\alpha}\right) = \sqrt{k}\left(d_\beta(\hat{\xi}_n, \hat{\sigma}_n/a(n/k)) - d_\beta(\xi, 1)\right) \\- \sqrt{k}\left(\frac{a(\tau_u)/a(n/k) - 1}{1-{\xi}}\left(1+ \frac{s_{u,\alpha}^{{\xi}} - 1}{{\xi}}\right) + \frac{s_{u,\alpha}^{{\xi}} -\beta^\xi}{\xi (1-\xi)}\right) = S -R,
\end{multline}
where we have denoted each term by $S$ and $R$, respectively.
By the delta method (\cref{delta_method}),
\[
S \dto N(0, \nabla d_{\beta}(\xi, 1)^\top \mathbf{\Sigma} \nabla d_{\beta}(\xi, 1)).
\] For the second term, let
\[
P = \sqrt{k}\left(\frac{a(\tau_u)}{a(n/k)} - (k\tau_u/n)^\xi\right) \quad \textnormal{ and } \quad Q = \sqrt{k}\left((k\tau_u/n)^\xi - 1\right).
\]
Then $\sqrt{k}(a(\tau_u)/a(n/k) - 1) = P+Q$ and from \cref{su_beta_form}, $\sqrt{k}(s_{u,\alpha}^\xi -\beta^\xi) = -s_{u,\alpha}^\xi Q$. Hence,
\[
R = \frac{1}{\xi(1-\xi)}\left[(P+Q)(\xi+s_{u,\alpha}^\xi-1) -s_{u,\alpha}^\xi Q\right]
= \frac{(\xi+s_{u,\alpha}^\xi-1)}{\xi(1-\xi)}P - \frac{Q}{\xi}.
\]
\cref{su_corr} implies that $s_{u,\alpha}^\xi=\beta^\xi +o_p(1)$ as $n\to\infty$ and from \cref{ank_lemma} we know that $P=o_p(1)$ as $n\to\infty.$ By \cref{tau_lemma}, $Q\dto N(0, \xi^2)$, and so $R\dto N(0,1).$ By the assumption that $(\hat{\xi}_n, \hat{\sigma}_n)$ and $Q$ are asymptotically independent, $R$ and $S$ are as well. Hence, \cref{delta_plus_remainder} implies 
\[
\frac{\sqrt{k}}{a(n/k)}\left(\POT - c_{u,\alpha}\right) \dto N\left(0, \nabla d_{\beta}(\xi, 1)^\top \mathbf{\Sigma} \nabla d_{\beta}(\xi, 1) + 1\right) = N(0, V),
\]
where
\[
\nabla d_{\beta}(\xi, 1) = 
\left[\
\frac{\partial d_{\beta}}{\partial x}(\xi, 1),
\frac{\partial d_{\beta}}{\partial y}(\xi, 1)
\right]^\top,
\]
with
\[
\frac{\partial d_{\beta}}{\partial x}(\xi, 1) = \frac{\beta^\xi(2\xi + \xi(1-\xi)\log\beta -1)}{\xi^2(1-\xi)^2} + \frac{1}{\xi^2}, \quad
\frac{\partial d_{\beta}}{\partial y}(\xi, 1) = \frac{\beta^{\xi} +\xi - 1}{\xi(1-\xi)}.
\]
\qed

\subsection{Proof of Theorem 5.1}
We first prove the following lemma, which shows that when $\tau_u$ is replaced by $n/k$ in \cref{gpd_approximation_error}, the asymptotic behaviour of $\epsilon_{u,\alpha}$ is the same (in probability).
\begin{lemma}\label{gpd_approx_error_emp}
Suppose that the assumptions of theorem 4.1 hold. Let $\alpha=1-(1/\beta)k/n$ where $\beta > 1$ is a constant not depending on $n$. Then,
\[
\frac{\Egpd}{a(n/k)A(n/k)K_{\xi, \rho}(\beta)} \pto 1.
\]
\begin{proof}
We follow the same line of reasoning as in the proof of \cref{gpd_approximation_error}. First, we have
\begin{align*}
q_{u,\alpha} &= u + \frac{\sigma(u)}{\xi}(s_{u,\alpha}^\xi - 1)\\
&= U(\tau_u) + \frac{a(\tau_u)}{\xi}(s_{u,\alpha}^\xi - 1)\\
&= U(n/k) + \frac{a(n/k)}{\xi}(\beta^\xi - 1) + U(\tau_u) - U(n/k) + \frac{1}{\xi}\left[a(\tau_u)(s_{u,\alpha}^\xi - 1) - a(n/k)(\beta^\xi - 1)\right] \\
&= U(n/k) + \frac{a(n/k)}{\xi}(\beta^\xi - 1) + U(\tau_u) - U(n/k) \\ &\qquad\qquad\qquad\qquad\qquad\qquad\;\;+ \frac{1}{\xi}\left[a(\tau_u)(s_{u,\alpha}^\xi - \beta^\xi) + (a(\tau_u) - a(n/k))(\beta^\xi - 1)\right].
\end{align*}
Given that $q_\alpha = U(1/(1-\alpha)) = U(n\beta/k)$,
\begin{align*}\label{q_expanded}
\frac{q_\alpha - q_{u,\alpha}}{a(n/k)A(n/k)} &= \frac{\frac{U(n\beta/k)-U(n/k)}{a(n/k)} - \frac{\beta^\xi-1}{\xi}}{A(n/k)} \\ &\quad -\frac{U(\tau_u) - U(n/k)}{a(n/k)A(n/k)} - \frac{a(\tau_u)(s_{u,\alpha}^\xi - \beta^\xi)}{\xi a(n/k) A(n/k)} - \left(\frac{a(\tau_u)}{a(n/k)}-1\right)\frac{\beta^\xi-1}{\xi A(n/k)} \\
&= I - II - III - IV.
\end{align*}
In what follows, terms $I$-$IV$ will be analyzed separately then finally combined.

$I$: By \cref{I_limit_cor}, with $t=n/k$ and $x=\beta$ we know that term $I$ tends to $I_{\xi,\rho}(\beta)$ as $n\to\infty.$ 

$II$: Under \cref{F_assumption} and \cref{A_conditions}, the following uniform inequality from \citet[Theorem 2.3.6]{dehaan2006} holds: for any $\varepsilon, \delta > 0$ there exists $t_0 = t_0(\varepsilon, \delta)$ such that for all $t, tx \ge t_0$,
\begin{equation} \label{U_uniform}
\left|\frac{\frac{U(tx)-U(t)}{a(t)}-\frac{x^\xi-1}{\xi}}{A(t)}-\frac{x^{\xi+\rho}-1}{\xi+\rho}\right| \leq \varepsilon x^{\xi+\rho} \max \left(x^{\delta}, x^{-\delta}\right).
\end{equation}
We can write $II$ as
\[
II = \frac{\frac{U(\tau_u)-U(n/k)}{a(n/k)}-\frac{(k\tau_u/n)^\xi-1}{\xi}}{A(n/k)} + \frac{(k\tau_u/n)^\xi-1}{\xi A(n/k)}.
\]
Hence, with $t=n/k$ and $x=k\tau_u/n$, the first term tends to $0$ in probability using \cref{U_uniform} and essentially the same arguments as in the proof of \cref{ank_lemma}. So, by 
the assumption of \cref{mle_normality} that $\sqrt{k}A(n/k) \to \lambda < \infty$ as $n\to\infty$ and \cref{tau_lemma},
\[
II = \frac{Q}{\xi \sqrt{k} A(n/k)} + o_p(1) = \frac{Z}{\lambda} + o_p(1), \quad n\to\infty,
\]
where $Q = \sqrt{k}((k\tau_u/n)^\xi - 1)$ and $Z$ denotes a standard normal random variable.

$III$: \cref{ank_corr} implies that $a(\tau_u)/a(n/k) \pto 1$ and by \cref{su_beta_form}, $s_{u,\alpha}^\xi - \beta^\xi = -s_{u,\alpha}^\xi Q/\sqrt{k}$. \Cref{su_corr} implies that $s_{u,\alpha}^\xi\pto\beta^\xi$, and so
\[
III = \frac{a(\tau_u)(s_{u,\alpha}^\xi - \beta^\xi)}{\xi a(n/k) A(n/k)} =  \frac{-s_{u,\alpha}^\xi Q}{\xi \sqrt{k} A(n/k)}(1+o_p(1)) = -\frac{\beta^\xi}{\lambda}Z + o_p(1), \quad n\to\infty.
\]

$IV$: With $P=\sqrt{k}\left(\frac{a(\tau_u)}{a(n/k)} - (k\tau_u/n)^\xi\right)$ and applying \cref{ank_lemma},
\[
IV = \left(\frac{a(\tau_u)}{a(n/k)}-1\right)\frac{\beta^\xi-1}{\xi A(n/k)} = (P+Q)\frac{\beta^\xi-1}{\xi \sqrt{k} A(n/k)} = \frac{\beta^\xi-1}{\lambda}Z + o_p(1), \quad n\to\infty.
\]
Now combining all terms, as $n\to\infty$,
\begin{multline*}
\frac{q_\alpha - q_{u,\alpha}}{a(n/k)A(n/k)} = I - II - III - IV \\= I_{\xi,\rho}(\beta) - \frac{Z}{\lambda} + \frac{\beta^\xi }{\lambda}Z - \frac{\beta^\xi-1}{\lambda}Z + o_p(1) = I_{\xi,\rho}(\beta) +o_p(1), \quad n\to\infty.
\end{multline*}
Hence, following the same reasoning as in the proof of \cref{gpd_approximation_error},
\begin{multline*}
\frac{\epsilon_{u,\alpha}}{a(n/k)A(n/k)} = \frac{c_{u,\alpha} - c_{\alpha}}{a(n/k)A(n/k)} \\= -\frac{1}{1-\alpha}\int_\alpha^1 \frac{q_{\gamma}-q_{u,\gamma}}{a(n/k)A(n/k)} d\gamma \pto -\beta\int_\beta^\infty \frac{I_{\xi,\rho}(x)}{x^2} dx
= K_{\xi,\rho}(\beta).
\end{multline*}
\end{proof}
\end{lemma}

\vspace{0.5cm}

\textbf{Proof of Theorem 5.1.}
First,
\[
\hat{c}^{(n)}_{\epsilon,\alpha} - c_\alpha =
\hat{c}^{(n)}_{\alpha} - \hEgpd - c_\alpha = 
\hat{c}^{(n)}_{\alpha} -c_{u,\alpha}- \hEgpd + c_{u,\alpha} - c_\alpha = 
\hat{c}^{(n)}_{\alpha} -c_{u,\alpha}- \hEgpd + \epsilon_{u,\alpha}.
\]
Hence,
\begin{equation}\label{conf_int_terms}
\frac{\sqrt{k}(\hat{c}^{(n)}_{\epsilon,\alpha} - c_\alpha)}{\hat{\sigma}_n} = 
\frac{\sqrt{k}(\hat{c}^{(n)}_{\alpha} -c_{u,\alpha})}{\hat{\sigma}_n} - \frac{\sqrt{k}(\hEgpd - \epsilon_{u,\alpha})}{\hat{\sigma}_n}.
\end{equation}
For the first term on the right-hand side of \cref{conf_int_terms},
\begin{equation}\label{term1_convergence}
\frac{\sqrt{k}(\hat{c}^{(n)}_{\alpha} -c_{u,\alpha})}{\hat{\sigma}_n} =
\frac{a(n/k)}{\hat{\sigma}_n}\frac{\sqrt{k}(\hat{c}^{(n)}_{\alpha} -c_{u,\alpha})}{a(n/k)}
\dto N(0,V),
\end{equation}
which follows from \cref{thm:pot_normality} and applying \cref{slutsky} with the fact that $\hat{\sigma}_n/a(n/k)\pto1$.

For the second term, first recall that
\[
\frac{\hEgpd}{a(n/k)A(n/k)} = \frac{\hat{\sigma}_n\hat{A}_n\hat{K}_n}{a(n/k)A(n/k)} \pto K_{\xi, \rho}(\beta),
\]
which follows from \cref{slutsky} and the continuous mapping theorem. Then, under the assumption that $\sqrt{k}A(n/k) \to \lambda < \infty$ ($n\to\infty$), it follows from \cref{slutsky} and \cref{gpd_approx_error_emp} that
\begin{multline}\label{term2_convergence}
\frac{\sqrt{k}(\hEgpd - \epsilon_{u,\alpha})}{\hat{\sigma}_n}
=\frac{a(n/k)\sqrt{k}A(n/k)}{\hat{\sigma}_n} \left(\frac{\hEgpd - \epsilon_{u,\alpha}}{a(n/k)A(n/k)}\right) \\=  \lambda (1+o_p(1)) \left(K_{\xi,\rho}(\beta) - K_{\xi, \rho}(\beta) + o_p(1)\right) = o_p(1), \quad n\to\infty.
\end{multline}
Combining the convergence in \cref{term1_convergence} and \cref{term2_convergence} with \cref{conf_int_terms}, it follows that
\[
\frac{\sqrt{k}(\hat{c}^{(n)}_{\epsilon,\alpha} - c_\alpha)}{\hat{\sigma}_n} \dto N(0,V),
\]
and hence, 
\[
\frac{\sqrt{k}(\hat{c}^{(n)}_{\epsilon,\alpha} - c_\alpha)}{\hat{\sigma}_n\sqrt{\hat{V}_n}}  = 
\frac{\sqrt{k}(\hat{c}^{(n)}_{\epsilon,\alpha} - c_\alpha)}{\hat{\sigma}_n\sqrt{V}} \frac{\sqrt{V}}{\sqrt{\hat{V}_n}}\dto N(0,1),
\]
which follows from the fact that $\hat{V}_n \pto V$ (from the continuous mapping theorem) and \cref{slutsky}.
\qed

\subsection{Consistency of \textit{A(n/k)} Estimator}\label{A_consistency}
In \cite{nawel}, an estimator for $A_0(n/k)$ is given,\footnote{The results of \cite{nawel} are presented in the truncated data setting, where for a sample $(X_i, Y_i), \, i=1,\ldots,n$ from a couple of independent random variables $(X,Y)$, $X_i$ is only observed when $X_i \le Y_i$. Their results can be adapted to the non-truncation setting by assuming that $\Prob(X\le Y) = 1$.} where the function $A_0$ satisfies the second-order condition of \citet[Theorem 2.3.9]{dehaan2006}, where for all $x>0,$
\begin{equation} \label{second_order_condition}
\lim _{t \rightarrow \infty} \frac{\frac{U(t x)}{U(t)}-x^\xi}{A_0(t)}=x^{\xi} \frac{x^{\rho}-1}{\rho}.
\end{equation}
Note that under \cref{F_assumption} and \cref{A_conditions}, \cref{second_order_condition} is satisfied. The relation between the function $A$ defined in \cref{A_equation} and $A_0$ is given in \citet[Table 3.1]{dehaan2006}, where
\begin{equation} \label{A_relation}
A = \frac{\xi +\rho}{\xi} A_0.
\end{equation}
We shall use this relation and an estimator for $A_0(n/k)$ to derive an estimator for $A(n/k)$. To prove consistency of the forthcoming estimator, we start with the following relation from \cite{nawel}:
\begin{equation} \label{A0_conv}
\lim_{t\to\infty}\frac{A_0(t)}{R(t)} = 1,
\end{equation}
where 
\[
R(t) = \frac{(1-\rho^2)(M^{(2)}(t) - 2(M^{(1)}(t))^2)}{2\rho M^{(1)}(t)}, \quad M^{(j)}(t) = t \int_{U(t)}^{\infty} \log ^{j}\left(x / U(t)\right) d F(x).
\]
This leads to an estimator for $A_0(n/k)$ \citet[p. 7]{nawel},
\[
\hat{A}_0^{(n)} = \frac{(1-\hat{\rho}_n^2)(\hat{M}_n^{(2)} - 2(\hat{M}_n^{(1)})^2)}{2\hat{\rho}_n \hat{M}_n^{(1)}},
\]
where $\hat{M}_n^{(j)}$ is an estimator for $M^{(j)}(n/k)$, given by 
\[
\hat{M}_n^{(j)} = \frac{1}{k}\sum_{i=1}^k[\log X_{(n-i+1,k)} - \log \thresh]^j,
\]
which is also given in \cref{se:second-order-A}.
Note that $\hat{M}_n^{(1)}$ is the well-known Hill estimator of $\xi$. $\hat{M}_n^{(j)}$ is consistent for $j=1,2$ under the conditions of the following lemma.
\begin{lemma} \label{M_lemma}
Suppose that \cref{F_assumption} holds. If $k = k_n \to \infty, \, k/n \to 0$ as $n\to\infty$, then
\[
\frac{\hat{M}_n^{(j)}}{M^{(j)}(n/k)} \pto 1, \qquad j=1,2.
\]
\begin{proof}
From \citet[equation 1.9]{nawel},
\[
M^{(1)}(t) \to \xi \quad \textnormal{ and } \quad M^{(2)}(t) \to 2\xi^2, \quad t\to\infty.
\]
By \citet[Theorem 3.2.2]{dehaan2006}, $\hat{M}_n^{(1)} \pto \xi$, and by \citet[Equation 3.5.7]{dehaan2006}, $\hat{M}_n^{(2)} \pto 2\xi^2$. Hence, by \cref{slutsky},
\[
\frac{\hat{M}_n^{(1)}}{M^{(1)}(n/k)} \to 1 \quad \textnormal{ and } \quad \frac{\hat{M}_n^{(2)}}{M^{(2)}(n/k)} \to 1, \quad n\to\infty.
\]
\end{proof}
\end{lemma}
From \cref{A_relation}, an estimator for $A(n/k)$ is 
\[
\hat{A}_n \triangleq \frac{(\hXi+\hRho)(1-\hRho)^2(\hat{M}_n^{(2)} - 2(\hat{M}_n^{(1)})^2)}{2\hXi\hRho \hat{M}_n^{(1)}},
\]
which is consistent under the conditions of the following lemma.
\begin{lemma}
Suppose that the assumptions of \cref{mle_normality} hold. If $\hat{\rho}_n \pto \rho$,
\[
\frac{\hat{A}_n}{A(n/k)} \pto 1.
\]
\end{lemma}
\begin{proof}
By theorem 4.1, $\hXi \pto \xi$, and so by \cref{A0_conv} and \cref{A_relation},
\[
\frac{\hat{A}_n}{A(n/k)} = \frac{\hXi+\hRho}{\hXi}\cdot\frac{\xi}{\xi+\rho}\cdot \frac{\hat{A}_0^{(n)}}{A_0(n/k)}
= (1+o_p(1))\frac{\hat{A}_0^{(n)}R(n/k)}{R(n/k)A_0(n/k)} =
(1+o_p(1))\frac{\hat{A}_0^{(n)}}{R(n/k)},
\]
as $n\to\infty$. By \cref{M_lemma},
\begin{align*}
\frac{\hat{A}_0^{(n)}}{R(n/k)} &= \frac{\rho (1-\hat{\rho}_n^2)  M^{(1)}(n/k)}{\hat{\rho}_n (1-\rho^2)\hat{M}_n^{(1)}}\cdot \frac{\hat{M}_n^{(2)} - 2(\hat{M}_n^{(1)})^2}{M^{(2)}(n/k) - 2(M^{(1)}(n/k))^2} \\ &= (1+o_p(1))\frac{\hat{M}_n^{(2)} - 2(\hat{M}_n^{(1)})^2}{M^{(2)}(n/k) - 2(M^{(1)}(n/k))^2},
\end{align*}
as $n\to\infty$. From \citet[p. 389]{Gomes2002Semiparametric}, 
\[
\frac{\hat{M}_n^{(2)} - 2(\hat{M}_n^{(1)})^2}{A_0(n/k)} \pto \frac{2\xi\rho}{(1-\rho)^2}.
\]
Hence,
\[
\frac{\hat{M}_n^{(2)} - 2(\hat{M}_n^{(1)})^2}{M^{(2)}(n/k) - 2(M^{(1)}(n/k))^2} = (1+o_p(1)) \frac{2\xi\rho}{(1-\rho)^2} \cdot \frac{A_0(n/k)}{M^{(2)}(n/k) - 2(M^{(1)}(n/k))^2}, \quad n\to\infty.
\]
Finally, combining \cref{A0_conv} with the fact that $M^{(1)}(n/k) \to \xi$ as $n\to\infty,$
\[
\frac{A_0(n/k)}{M^{(2)}(n/k) - 2(M^{(1)}(n/k))^2} = (1+o(1))\frac{(1-\rho)^2}{2\rho M^{(1)}(n/k)} = (1+o(1))\frac{(1-\rho)^2}{2\rho \xi}, \quad n\to\infty,
\]
and thus
\[
\frac{\hat{M}_n^{(2)} - 2(\hat{M}_n^{(1)})^2}{M^{(2)}(n/k) - 2(M^{(1)}(n/k))^2} \pto 1 \qquad \textnormal{which implies} \qquad \frac{\hat{A}_n}{A(n/k)} \pto 1.
\]
\end{proof}

\section{Estimation Algorithms}
\subsection{Adaptive \texorpdfstring{$\rho$}{rho} Estimation}\label{ada_rho}
The $\rho$ estimator given in \cref{se:second-order-rho},
\[
\hat{\rho}_n = \frac{3 (T_n^{(\tau)}(m) - 1)}{T_n^{(\tau)}(m)-3},
\]
requires the choice of two parameters: a sample fraction $m$, and tuning parameter $\tau$. Depending on the underlying distribution, the reliability of $\hat{\rho}_n$ can be very sensitive to the choice $m$ and $\tau$. The adaptive algorithm of \cite[Section 4.1]{bias_reduce_rho} provides an automated way to select these parameters. We present a slightly modified version of their algorithm here, which we use in our experiments. 

\begin{algorithm}[H]
\SetKwInput{Input}{Input}
\SetKwInput{Output}{Output}
\Indm
  \Input{An i.i.d. sample $X_1,..,X_n$, test parameters $\tau_1,\ldots\tau_q$, test sample fractions $m_1,\ldots,m_r$, precision $p$.}
   \Output{$\hat{\rho}_n$}
\Indp
\BlankLine
\For{$i=1,\ldots q$}{
\For{$j=1,\ldots r$}{
Compute $\hat{\rho}^{(\tau_i)}_n(m_j)$ using \cref{rho_estimator}, rounded to $p$ decimal places
}
Set $m^{(\tau_i)}_\textnormal{min}$, $m^{(\tau_i)}_\textnormal{max}$ to be the minimum and maximum $m$ values associated with the longest run of consecutive equal $\hat{\rho}^{(\tau_i)}_n$ values\\
Set $l^{(\tau_i)} = m^{(\tau_i)}_\textnormal{max} - m^{(\tau_i)}_\textnormal{min} + 1$, the length of the largest run
}
Set $k = \argmax_{i=1,\ldots,q}l^{(\tau_i)}$\\
Set $\hat{\rho}_n$ to the median of $\hat{\rho}^{(\tau_k)}_n(m^{(\tau_k)}_\textnormal{min}),\hat{\rho}^{(\tau_k)}_n(m^{(\tau_k)}_\textnormal{min}+1),\ldots,\hat{\rho}^{(\tau_k)}_n(m^{(\tau_k)}_\textnormal{max})$

\caption{Adaptive algorithm for $\rho$ estimation (ADARHO)}
\label{algo:ada_rho}
\end{algorithm}

\subsection{Automated Threshold Selection}\label{auto_thresh}
The method of \cite{bader2018automated} is as follows. Consider a
fixed set of thresholds $u_1 < \ldots < u_l$, where for each $u_i$ we
have $k_i$ excesses.  The sequence of null hypotheses for each
respective test $i$, $i=1,\ldots,l$, is given by
\begin{center}
\begin{tabular}{cc}
$H_0^{(i)}:$ & The distribution of the $k_i$ excesses above $u_i$ follows the GPD.
\end{tabular}
\end{center}
For each threshold $u_i$, let $\hat{\theta}_i = (\hat{\xi}^{(n)}_{u_i}, \hat{\sigma}^{(n)}_{u_i})$ denote the MLEs computed from the $k_i$ excesses above $u_i$. The Anderson-Darling (AD) test statistic
comparing the empirical threshold excesses distribution with the GPD
is then calculated. Let $y_{(1)} < ... < y_{(k_i)}$ denote the ordered threshold excesses for test $i$, and apply the transformation $z_{(j)} = G_{\hat{\theta}_i}(y_{(j)})$, $j=1,\ldots k_i$, where $G$ denotes the cdf of the GPD. The AD statistic for test $i$ is then
\begin{equation} \label{ad_statistic}
A_{i}^{2}=-k_i-\frac{1}{k_i} \sum_{j=1}^{k_i}(2 j-1)\left[\log \left(z_{(j)}\right)+\log \left(1-z_{(k_i+1-j)}\right)\right].
\end{equation}
Corresponding $p$-values for each test statistic can then be found by
referring to a lookup table (e.g., \citet{stephens}) or computed on-the-fly. Using the $p$-values $p_1, \ldots, p_l$ calculated for each test, the ForwardStop
rule of \cite{gsell} is used to choose the threshold. This is done by calculating
\begin{equation}
\label{BaderIndex}
\hat{w}_F=\max \bigg\{w \in I \, \bigg\vert \, -\frac{1}{w} \sum_{i=1}^{w} \log \left(1-p_{i}\right) \leq \gamma\bigg\},
\end{equation}
where $\gamma$ is a chosen significance parameter and $I\subseteq \{1,\ldots,l\}$, $I\neq \emptyset$. Under this rule, the threshold $u_v$ is chosen, where $v=\min\{w\in I \, \vert \, w > \hat{w}_F\}$. If no $\hat{w}_F$ exists, then no rejection is made and $u_{\min(I)}$ is chosen. If $\hat{w}_F = \max(I)$, then $u_{\max(I)}$ is chosen. The overall procedure is summarized in \Cref{algo:auto_thresh}.
\begin{remark}
In the threshold selection procedure of \cite{bader2018automated}, $\hat{w}_F$ is given with $I = \{1,\ldots,l\}$, but we make the modification that $I$ is an arbitrary index set in view of CVaR estimation: since $c_{u,\alpha}$ tends to infinity when $\xi$ tends to $1$, in order to ensure reasonable estimates of the CVaR we use a cutoff parameter $\xi_{max} < 1$, where the MLE $\hat{\xi}^{(n)}_{u_i}$ and corresponding threshold $u_i$ are discarded if $\hat{\xi}^{(n)}_{u_i} > \xi_{max}$.
\end{remark}
\begin{remark}
Instead of choosing the candidate thresholds $u_1,\ldots,u_l$ directly, it is usually more convenient to choose \textit{threshold percentiles} $q_1,\ldots, q_l$ and compute $u$ values via the empirical quantile function, i.e., $u_i = \hat{F}^{-1}_n(q_i)$.  
\end{remark}

\begin{algorithm}[H]
\SetKwInput{Input}{Input}
\SetKwInput{Output}{Output}
\Indm
  \Input{An i.i.d. sample $X_1,..,X_n$, significance parameter $\gamma$, threshold percentiles  $0 < q_1,\ldots, q_l < 1$, cutoff $\xi_{max} < 1$.}
   \Output{$(\hXi, \hSig), u$ if $I\neq \emptyset$, otherwise return NaN}
\Indp
\BlankLine
$I \leftarrow \emptyset$ \\
\For{$i=1,\ldots l$}{
Set $u_i = \hat{F}^{-1}_n(q_i)$\\
Compute $(\hat{\xi}^{(n)}_{u_i}, \hat{\sigma}^{(n)}_{u_i})$ from $k_i$ threshold excesses using maximum likelihood\\
\If{$\hat{\xi}^{(n)}_{u_i} \leq \xi_{max}$}{
Compute $A_i^2$ using \cref{ad_statistic}\\
Set $p_i$ to $p$-value for $A_i^2$ using lookup table\\
$I \leftarrow I \cup \{i\}$
}
}
\If{$I \neq \emptyset$}{
Set $W = \{w \in I \, \vert \, -\frac{1}{w} \sum_{i=1}^{w} \log \left(1-p_{i}\right) \leq \gamma\}$\\
\eIf{$W \neq \emptyset$}{
Compute $\hat{w}_F$ using \cref{BaderIndex} \\
\eIf{$\hat{w}_F=\max(I)$} {
$v \leftarrow \max(I)$
} 
{
$v \leftarrow \min\{w\in I \, \vert \, w > \hat{w}_F\}$
}
}
{
$v \leftarrow \min(I)$\\
}
$u \leftarrow u_v$\\
$(\hXi, \hSig) \leftarrow (\hat{\xi}^{(n)}_{u_v}, \hat{\sigma}^{(n)}_{u_v})$
}
\caption{Automated threshold selection (AUTOTHRESH)}
\label{algo:auto_thresh}
\end{algorithm}

\subsection{Algorithm to Compute the Unbiased POT Estimator}\label{upot_algo}
This section provides the algorithm used to compute UPOT in its entirety, which makes use of both \cref{algo:ada_rho} and \cref{algo:auto_thresh}. In our experiments, we set $\tau_1=-1.5, \tau_2=-1.25, \ldots, \tau_{13}=1.5$, $m_1=100,m_2=200,\ldots,m_r=n-1$, $p=1$ in \cref{algo:ada_rho}, and $\gamma=0.1$, $q_1=0.79,q_2=0.80,\ldots,q_{20}=0.98$, $\xi_{max} = 0.9$ in \cref{algo:auto_thresh}. Assume these choices of values in the following algorithm.
\begin{algorithm}[H]
\SetKwInput{Input}{Input}
\SetKwInput{Output}{Output}
\Indm
  \Input{An i.i.d. sample $X_1,..,X_n$, confidence level $\alpha$}
  \Output{$\UPOT$}
\Indp
\BlankLine
$\mathbf{x} \leftarrow$ AUTOTHRESH($X_1,..,X_n$) \\
\If{$\mathbf{x} \textnormal{ is not NaN}$} {
$(\hXi, \hSig), u \leftarrow \mathbf{x}$ \\
$\hat{\rho_n} \leftarrow$ ADARHO($X_1,..,X_n$) \\
Compute $\hat{b}_n$ using \cref{bias_estimator} \\
Compute $\hat{A}_n$ using \cref{A_estimator} and the $k$ threshold excesses above $u$\\
$(\hat{\xi}_n, \hat{\sigma}_n) \leftarrow (\hXi - \hat{A}_n\hat{b}^{(1)}_n, \;  \hSig(1-\hat{A}_n\hat{b}^{(2)}_n))$ \\
Compute $\hEgpd$ using \cref{approx_error_estimator} \\
Compute $\POT$ using \cref{pot_estimator} \\
$\UPOT \leftarrow \POT - \hEgpd$\\
}
\caption{Unbiased peaks-over-threshold CVaR estimator (UPOT)} \label{algo:UPOT}
\end{algorithm}

\begin{remark}
It may happen that \Cref{algo:UPOT} fails if AUTOTHRESH returns NaN, in which case no suitable estimates of $\xi$ are found. This is an indication that the underlying data distribution does not satisfy the condition $\xi < 1$ and the CVaR does not exist. To make \Cref{algo:UPOT} robust, the sample average estimate is used as a fallback when the latter occurs. We report the failure rate of UPOT during experiments in \cref{table1}.
\end{remark}

\section{Examples of Heavy-tailed Distributions} \label{appendix:dists}
\subsection{Burr}
The Burr distribution with parameters $c, d$ has cdf given by
\[
F_{c,d}(x)=1-\left(1+x^{c}\right)^{-d}, \quad c, d, x > 0.
\]
The CVaR for the Burr distribution can be derived from its expression  for the conditional moment given in \citet[Section 2.2]{burr_properties}. If $X \sim \textrm{Burr}(c, d)$,
\begin{equation}\label{burr_cvar}
\cvar(X) = \frac{d[(1/q_\alpha)^c]^{d-1/c}}{(1-\alpha)(d-1/c)} \,_2F_1\left(d-\frac{1}{c}, 1+d, d-\frac{1}{c}+1; -\frac{1}{q_\alpha}\right), \quad cd > 1,
\end{equation}
where $\,_2F_1$ denotes the hypergeometric function and $q_\alpha = F^{-1}_{c,d}(\alpha)$. Values of $\xi, \rho$ and functions $a$ and $A$ are given by
$$
\xi=\frac{1}{cd}, \quad \rho=-\frac{1}{d}, \quad a(t) = \frac{t^{1/d}}{cd}\left(t^{1/d}-1\right)^{1/c-1}, \quad A(t) = \frac{1-c}{cd(t^{1/d}-1)},
$$
where $a$ and $A$ are defined for $t\ge1$.

\subsection{Fr\'echet} \label{se:frechet}
The Fr\'echet distribution with parameter $\gamma$ has cdf given by
\[
F_\gamma(x)=e^{-x^{-\gamma}}, \quad \gamma, x > 0.
\]
If $X \sim \textrm{Fr\'echet}(\gamma)$,
\begin{equation}\label{frec_cvar}
\cvar(X) = (1-\alpha)^{-1}\left[\Gamma\left({\gamma-1}/{\gamma}\right) \right.\\- \left.\Gamma\left({\gamma-1}/{\gamma}, -\log(\alpha)\right)\right], \quad \gamma > 1,
\end{equation}
where $\Gamma(\cdot)$ and $\Gamma(\cdot,\cdot)$ denote the gamma and upper incomplete gamma functions, respectively. Values of $\xi, \rho$ and functions $a$ and $A$ are given by
$$
\xi=\frac{1}{\gamma}, \quad \rho=-1, \quad a(t) = \frac{\log{\left(\frac{t}{t-1}\right)}^{-1-\frac{1}{\gamma}}}{\gamma(t-1)}, \quad A(t) = -\frac{1+\gamma+\gamma t \log(1-\frac{1}{t})}{\gamma(1-t)\log(1-\frac{1}{t})} - \frac{1}{\gamma},
$$
where $a$ and $A$ are defined for $t\ge1$.
\subsubsection{Asymptotic variance of SA estimator for Fr\'echet distribution}\label{asymp_sa_frec}
An expression for the asymptotic variance (AVAR) of the SA estimator is given in \cite[Theorem 2]{TRINDADE20073524}. Let $Z$ be a continuous random variable such that $\mathbb{E}[Z^2]$ is finite. Then, for a  confidence level $\alpha$,
\[
\sqrt{n}\left(\cvar(Z) - \widehat{\textnormal{CVaR}}_{n,\alpha}(Z)\right) \dto N(0, \theta^2),
\]
where $\widehat{\textnormal{CVaR}}_{n,\alpha}(Z)$ is the SA estimator given in \cref{sa_cvar} and 
\[
\theta^2 = \frac{\textnormal{Var}\left([Z-q_\alpha]^+\right)}{(1-\alpha)^2},
\]
and $[x]^+ = \max\{0, x\}$.
If $Z \sim \textrm{Fr\'echet}(\gamma)$, the condition that $\mathbb{E}[Z^2]$ is finite is equivalent to $\gamma > 2$. By the law of total expectation,
\[
\mathbb{E}[[Z-q_\alpha]^+] = \mathbb{P}(Z\le q_\alpha) \mathbb{E}[0] + \mathbb{P}(Z > q_\alpha)\mathbb{E}[Z-q_\alpha | Z > q_\alpha] = (1-e^{-q_\alpha^{-\gamma}})\mathbb{E}[Z-q_\alpha | Z > q_\alpha].
\]
The distribution of the conditional random variable on the right hand side has the same form as the excess distribution function, given in \cref{def:edf}. Let
\[
F_{\alpha,\gamma}(z) = \mathbb{P}(Z-q_\alpha \le z | Z > q_\alpha) = \frac{F_\gamma(z+q_\alpha) - F_\gamma(q_\alpha)}{1-F_\gamma(q_\alpha)} = \frac{e^{-(z+q_\alpha)^{-\gamma}}-e^{-q_\alpha^{-\gamma}}}{1-e^{-q_\alpha^{-\gamma}}},
\]
\[
f_{\alpha,\gamma}(z) = F^\prime_{\alpha,\gamma}(z) = \frac{\gamma (z+q_\alpha)^{-\gamma-1}e^{-(z+q_\alpha)^{-\gamma}}}{1-e^{-q_\alpha^{-\gamma}}}, \quad z>0.
\]
Hence,
\[
\mathbb{E}[[Z-q_\alpha]^+] = \int_0^\infty \gamma (z+q_\alpha)^{-\gamma-1} z e^{-(z+q_\alpha)^{-\gamma}} dz = \int_{q_\alpha}^\infty t^{-1/\gamma} e^{-t} dt = \Gamma(1-1/\gamma, q_\alpha), \quad \gamma >1,
\]
where $t=(z+q_\alpha)^{-\gamma}$ and $\Gamma(\cdot,\cdot)$ denotes the upper incomplete gamma function. With a similar calculation, the second moment is
\[
\mathbb{E}[([Z-q_\alpha]^+)^2] = \Gamma(1-2/\gamma, q_\alpha), \quad \gamma > 2.
\]
Finally, we can compute the AVAR of the SA estimator for the Fr\'echet distribution, which is
\[
\frac{\theta^2}{n} = \frac{\mathbb{E}[([Z-q_\alpha]^+)^2] - \mathbb{E}[[Z-q_\alpha]^+]^2}{n(1-\alpha)^2} = \frac{\Gamma(1-2/\gamma, q_\alpha) - \Gamma(1-1/\gamma, q_\alpha)^2}{n(1-\alpha)^2}, \quad \gamma > 2.
\]
\subsection{Half-\textit{t}}
If $X$ follows the $t$ distribution with $\nu$ degrees of freedom, then $|X|$ follows the half-$t$ distribution, which has cdf given by
\[
F_\nu (x)= 2-\mathcal{I}_{t(x)}\left(\frac{\nu}{2}, \frac{1}{2}\right), \quad \nu > 0, x \ge 0,
\]
where $t(x) = \frac{\nu}{x^2 + \nu}$ and $\mathcal{I}_t(a,b)$ is the regularized incomplete Beta function. The CVaR for the half-$t$ distribution can be derived from the expression for the CVaR of the $t$-distribution given in \cite[Proposition 12]{cvarfordists}. If $X\sim \textnormal{half-}t(\nu)$, then 
\[
\cvar(X) = 2\frac{\nu + q_\alpha}{(\nu-1)(1-\alpha)}g_\nu(q_\alpha), \quad \nu > 1,
\]
where $g_\nu$ is the probability density function of the standardized $t$-distribution, and $q_\alpha=T^{-1}\left(\frac{\alpha+1}{2}\right)$ where $T^{-1}$ is the inverse of the cdf of standardized $t$-distribution. The half-$t$ distribution is in $\textnormal{MDA}(H_\xi)$ with $\xi=1/\nu$, and has $\rho=-2/\nu$ (\citet[Remark 2.1]{bias_reduce_rho}). It does not seem possible to compute closed-form expressions for the functions $a$ and $A$ for the half-$t$ distribution.

\section{Numerical Results} \label{se:num_results}
\begin{table}[H] 
\caption{Data for all distributions used in experiments. $\cvar$ denotes the exact CVaR value for $\alpha=0.998$. Given at a sample size $n=50000$, UPOT, BPOT, and SA denote the average estimated CVaR values across $N=1000$ independent runs, and TP denotes the average threshold percentile chosen by \cref{algo:auto_thresh}. FR denotes the failure rate, the number of independent runs where \cref{algo:auto_thresh} returned NaN, i.e., where no suitable estimate of $\xi$ could be obtained and no CVaR estimate could be produced by the POT methods. This value is given at a sample size of $n=5000$ since very few failures occurred beyond this sample size. CP denotes the coverage probability achieved by our confidence interval at a sample size $n=50000$. }  
\centering
\begin{tabular}{lrrrrrrr} \\ \toprule
& $\cvar$ & UPOT & BPOT & SA & TP & FR & CP\\ \midrule
Burr(0.38, 4.0) & 124.87 & 89.83 & 235.70 & 121.03 & 0.96 & 2 & 0.73\\
Burr(0.5, 3.0) & 166.18 & 135.62 & 245.74 & 163.39 & 0.92 & 1 & 0.87\\
Burr(0.67, 2.25) & 175.93 & 140.53 & 219.55 & 173.19 & 0.84 & 0 & 0.88\\
Burr(2.0, 0.75) & 188.98 & 191.48 & 180.22 & 190.34 & 0.80 & 0 & 0.94\\
Burr(3.33, 0.45) & 190.15 & 189.71 & 187.26 & 192.83 & 0.80 & 4 & 0.95\\
Fr\'echet(1.5) & 188.96 & 188.94 & 182.11 & 181.76 & 0.80 & 4 & 0.89\\
Fr\'echet(1.75) & 81.32 & 81.88 & 78.91 & 81.97 & 0.80 & 2 & 0.93\\
Fr\'echet(2.0) & 44.71 & 44.76 & 43.25 & 44.52 & 0.80 & 1 & 0.94\\
Fr\'echet(2.25) & 28.49 & 28.66 & 27.75 & 28.45 & 0.80 & 2 & 0.95\\
Fr\'echet(2.5) & 20.02 & 20.07 & 19.49 & 19.94 & 0.80 & 1 & 0.95\\
half-$t$(1.5) & 156.58 & 159.92 & 145.89 & 175.91 & 0.81 & 0 & 0.94\\
half-$t$(1.75) & 74.52 & 75.40 & 69.49 & 74.02 & 0.82 & 2 & 0.94\\
half-$t$(2.0) & 44.70 & 45.36 & 42.29 & 44.50 & 0.83 & 0 & 0.94\\
half-$t$(2.25) & 30.74 & 31.27 & 29.26 & 30.68 & 0.84 & 1 & 0.95\\
half-$t$(2.5) & 23.10 & 23.34 & 22.15 & 23.12 & 0.85 & 1 & 0.94\\
\bottomrule  \label{table1} \end{tabular} \end{table}

\begin{table}[H]
\caption{Error values at sample size $n=50000.$}  
\centering
\begin{tabular}{lrrrrrr} \\ \toprule
& \multicolumn{3}{c}{RMSE} & \multicolumn{3}{c}{Bias}\\
\cmidrule(lr){2-4} \cmidrule(lr){5-7}
& UPOT & BPOT & SA & UPOT & BPOT & SA\\ \midrule
Burr(0.38, 4.0) & 48.56 & 134.15 & 64.04 & -35.03 & 110.83 & -3.84\\
Burr(0.5, 3.0) & 47.71 & 121.18 & 124.71 & -30.56 & 79.56 & -2.78\\
Burr(0.67, 2.25) & 48.88 & 58.97 & 81.34 & -35.41 & 43.62 & -2.75\\
Burr(2.0, 0.75) & 17.48 & 22.27 & 88.48 & 2.50 & -8.76 & 1.36\\
Burr(3.33, 0.45) & 13.83 & 19.40 & 128.88 & -0.44 & -2.89 & 2.67\\
Fr\'echet(1.5) & 19.47 & 21.31 & 69.35 & -0.02 & -6.85 & -7.19\\
Fr\'echet(1.75) & 6.10 & 7.07 & 24.25 & 0.56 & -2.41 & 0.65\\
Fr\'echet(2.0) & 2.71 & 3.36 & 7.45 & 0.05 & -1.47 & -0.20\\
Fr\'echet(2.25) & 1.50 & 1.90 & 3.21 & 0.16 & -0.75 & -0.05\\
Fr\'echet(2.5) & 0.92 & 1.18 & 1.69 & 0.05 & -0.52 & -0.08\\
half-$t$(1.5) & 16.78 & 22.68 & 765.05 & 3.34 & -10.69 & 19.33\\
half-$t$(1.75) & 6.11 & 8.72 & 16.40 & 0.89 & -5.03 & -0.50\\
half-$t$(2.0) & 3.58 & 4.92 & 7.62 & 0.66 & -2.41 & -0.20\\
half-$t$(2.25) & 2.07 & 2.78 & 3.49 & 0.53 & -1.48 & -0.06\\
half-$t$(2.5) & 1.44 & 1.88 & 2.06 & 0.23 & -0.95 & 0.02\\
\bottomrule \end{tabular} \end{table}

\end{document}